\documentclass[11pt]{article}
\usepackage{ltexpprt} 
\usepackage{fullpage}
\usepackage{amsmath,amsfonts,amssymb}
\usepackage[usenames]{color}

\usepackage{bm}

\usepackage{xspace}

\usepackage{paralist}

\usepackage{graphicx}
\usepackage{caption}
\usepackage{epstopdf}

\usepackage{boxedminipage}

\usepackage{subfig}

\usepackage[colorlinks,urlcolor=blue,citecolor=blue,linkcolor=blue]{hyperref}

\usepackage{subfig}
\usepackage{tabularx}
\usepackage{array}

\usepackage{booktabs}
\usepackage{textcase}

\usepackage{multirow}

\usepackage{braket}

\usepackage{tikz}
\usetikzlibrary{arrows,positioning,backgrounds}
\newcommand*\circled[1]{\tikz[baseline=(char.base)]{
            \node[shape=circle,draw,minimum size=3.5mm,inner sep=0pt] (char) {#1};}}

\usepackage{rotating}

\usepackage{algorithm}
\usepackage{algorithmic}
\usepackage{slashbox}

\newcommand{\cN}{{\cal N}}

\newcommand{\cV}{{\cal V}}

\newcommand{\qed}{\hfill $\Box$}

\newcommand{\cei}[1]{\lceil #1 \rceil}

\newcommand{\pr}{{\rm Pr}}
\newcommand{\Prob}[1]{\pr\left\{#1\right\}}

\newcommand{\eps}{\varepsilon}

\newcommand{\EX}{\hbox{\bf E}}

\newcommand{\Sec}[1]{\hyperref[sec:#1]{\S\ref*{sec:#1}}} %
\newcommand{\Eqn}[1]{\hyperref[eq:#1]{(\ref*{eq:#1})}} %
\newcommand{\Fig}[1]{\hyperref[fig:#1]{Fig.\,\ref*{fig:#1}}} %
\newcommand{\Tab}[1]{\hyperref[tab:#1]{Tab.\,\ref*{tab:#1}}} %
\newcommand{\Thm}[1]{\hyperref[thm:#1]{Thm.\,\ref*{thm:#1}}} %
\newcommand{\Lem}[1]{\hyperref[lem:#1]{Lem.\,\ref*{lem:#1}}} %
\newcommand{\Prop}[1]{\hyperref[prop:#1]{Prop.~\ref*{prop:#1}}} %
\newcommand{\Cor}[1]{\hyperref[cor:#1]{Cor.~\ref*{cor:#1}}} %
\newcommand{\Def}[1]{\hyperref[def:#1]{Defn.~\ref*{def:#1}}} %
\newcommand{\Alg}[1]{\hyperref[alg:#1]{Alg.\,\ref*{alg:#1}}} %
\newcommand{\Ex}[1]{\hyperref[ex:#1]{Ex.~\ref*{ex:#1}}} %
\newcommand{\Clm}[1]{\hyperref[clm:#1]{Claim~\ref*{clm:#1}}} %
\newcommand{\Step}[1]{\hyperref[step:#1]{Step~\ref*{step:#1}}} %

\newcommand{\dcc}{{C}}

\newcommand{\algkappa}{$\gcc$-{\tt wedge sampler}}
\newcommand{\alglcc}{$\lcc$-{\tt wedge sampler}}
\newcommand{\algdcc}{$\dcc_d$-{\tt wedge sampler}}
\newcommand{\algtrid}{$T_d$-{\tt wedge sampler}}
\newcommand{\algdtri}{$\gcc(\psi,\ttype)$-{\tt sampler}}

\newcommand{\din}[1]{d_{#1}^{\leftarrow}}
\newcommand{\dout}[1]{d_{#1}^{\rightarrow}}
\newcommand{\drec}[1]{d_{#1}^{\leftrightarrow}}

\newcommand{\gcc}{\kappa}
\newcommand{\lcc}{C}

\newcommand{\ttype}{\rho}

\newcommand{\comment}[1]{}

\newcommand{\Omit}[1]{}

\def\linkaba{stealth'-stealth'}
\def\linkab{-stealth'}
\def\linkba{stealth'-}

\newcommand{\directedtriangle}[4]{
  \node (A) at (0,0) [nd] {}; %
  \node (B) at +(0.76,0) [nd] {}; %
  \node (C) at +(0.38,0.65) [nd] {}; %
  \draw[#1] (A)--(B);
  \draw[#2] (A)--(C);
  \draw[#3] (B)--(C);
  \node [inner sep = 0,below] at +(0.5,-0.2) {\footnotesize #4};
}

\newcommand{\directedwedge}[3]{
  \node (A) at (0,0) [nd] {}; %
  \node (B) at +(0.7,0) [nd] {}; %
  \node (C) at +(0.35,0.6) [nd] {}; %
  \draw[#1] (A)--(C);
  \draw[#2] (B)--(C);
  \node [inner sep = 0,below] at +(0.5,-0.2) {\footnotesize #3};
}

\newcommand{\wout}{(i)\xspace}
\newcommand{\wmid}{(ii)\xspace}
\newcommand{\win}{(iii)\xspace}

\newcommand{\ttrans}{(a)\xspace}
\newcommand{\tcycle}{(b)\xspace}
\newcommand{\toutrecip}{(c)\xspace}
\newcommand{\tmidrecip}{(d)\xspace}

\tikzstyle{vertex}=[circle,draw=black,fill=teal!50,thin,inner sep=0pt,minimum size=6pt]

\def\typea{blue}
\def\typeb{red}
\def\typec{green}

\begin{document}

\title{Wedge Sampling for Computing Clustering Coefficients\\and Triangle Counts on Large Graphs 
\thanks{This manuscript is an extended version of \cite{SePiKo13}. This work was funded by the DARPA GRAPHS
    program and by the DOE ASCR Complex  Distributed Interconnected Systems
    (CDIS) program, and Sandia's Laboratory Directed Research \& Development (LDRD) program. Sandia National Laboratories is a multi-program
    laboratory managed and operated by Sandia Corporation, a wholly
    owned subsidiary of Lockheed Martin Corporation, for the
    U.S. Department of Energy's National Nuclear Security
    Administration under contract DE-AC04-94AL85000.}}
\date{}

\author{C. Seshadhri\thanks{Sandia National Laboratories, CA, scomand@sandia.gov} 
\and
Ali Pinar\thanks{Sandia National Laboratories, CA, apinar@sandia.gov} 
\and 
Tamara G. Kolda\thanks{Sandia National Laboratories, CA, tgkolda@sandia.gov} }

\maketitle

\begin{abstract} 
Graphs are used to model interactions in a variety of contexts, and 
there is a growing need to quickly assess the structure  of such graphs. 
Some of the most useful graph metrics
are based on \emph{triangles}, such as those measuring social cohesion.
Algorithms to compute them can be extremely expensive, even for moderately-sized graphs with only millions of edges.
Previous work has considered node and edge sampling; in contrast, we consider \emph{wedge sampling}, which provides faster and more accurate approximations than competing techniques.
Additionally, wedge sampling enables estimation local clustering coefficients, degree-wise clustering coefficients, uniform triangle sampling, 
and directed triangle counts.
Our methods come with provable and practical probabilistic error estimates  for all computations. We provide extensive results that show our methods are both more accurate and faster than state-of-the-art alternatives.
\end{abstract} 

\section{Introduction}
\label{sec:intro}
Graphs are used  to model  infrastructure networks,  the World Wide Web, computer traffic, molecular interactions, ecological systems, epidemics, citations, and social interactions, among others.  Despite the differences in the motivating applications, some topological structures have emerged to be important across all these domains. 
In particular, the \emph{triangle} is a manifestation of 
homophily (people become friends with those similar to themselves) and transitivity (friends of friends become friends).   
The triangle structure of a graph is commonly used in the social sciences for positing various theses
on behavior~\cite{Co88, Po98, Burt04, FoDeCo10}. 
Triangles have also been used  in graph mining applications such as spam detection and finding common topics on the WWW~\cite{EcMo02, BeBoCaGi08}.
A new generative model, Blocked Two-Level Erd\"os-R\'enyi (BTER) \cite{SeKoPi11}, can capture  triangle behavior in real-world graphs, but necessarily requires the degree-wise clustering coefficients as input.
Relationships among degrees of triangle vertices can also be used as a descriptor of the underlying graph~\cite{DuPiKo12}.

In this paper, we study the idea of \emph{wedge sampling}, i.e., choosing random wedges (from a uniform distribution over all wedges) to compute various triadic measures on large-scale graphs.
We provide precise probabilistic error bounds where the accuracy depends on the number of random samples.
The term \emph{wedge} refers to a path of length 2; in \Fig{example},
\circled{2}--\circled{1}--\circled{3} is a wedge \emph{centered} at node 1. A \emph{triangle} is a
cycle of length 3; in \Fig{example}, \circled{3}--\circled{4}--\circled{5} is a triangle. We say a
wedge is \emph{closed} if its vertices form a triangle. Observe that
each triangle consists of three closed wedges.

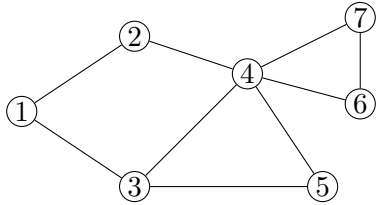
\begin{figure}[t!]
  \centering
  \begin{tikzpicture}[scale=0.5,
    nd/.style={circle,draw,fill=white,minimum size=4mm,inner sep=0pt}]
    \node (1) at (0,0) [nd] {1};
    \node (2) at (3,2) [nd] {2};
    \node (3) at (3,-2) [nd] {3};
    \node (4) at (6,1) [nd] {4};
    \node (5) at (8,-2) [nd] {5};
    \node (6) at (9,0.2) [nd] {6};
    \node (7) at (9,2.5) [nd] {7};
    \draw (1) to (2);
    \draw (1) to (3);
    \draw (2) to (4);
    \draw (3) to (4);
    \draw (3) to (5);
    \draw (4) to (5);
    \draw (4) to (6);
    \draw (4) to (7);
    \draw (6) to (7);
  \end{tikzpicture}
  \caption{Example graph with $n=7$ vertices, $m=8$ edges, $W=18$ wedges, and $T=2$ triangles. 
  }
  \label{fig:example}
\end{figure}
\subsection{Triadic measures on graphs} \label{sec:info}
 Let $n$, $m$, $W$, and $T$ denote
the number of nodes, edges, wedges, and triangles, respectively, in a graph $G$.
A standard algorithmic task is to count (or approximate) $T$.
There are many other ``triadic'' measures associated with graphs as well.
The following are classic quantities of interest, and are defined on \emph{undirected graphs}. (Formal definitions are given in \Tab{notation}.)

\medskip
\begin{asparaitem}
	\item Transitivity or global clustering coefficient~\cite{WaFa94}: This is $\kappa = 3T/W$, and is the
	fraction of wedges that are closed (i.e., participate in a triangle). Intuitively, it measures how often friends of
friends are also friends, and it is the most common triadic
measure. 
	\item Local clustering coefficient~\cite{WaSt98}: The clustering coefficient of
	vertex $v$ (denoted by $C_v$) is the fraction of wedges centered at $v$
	that are closed. The average of $C_v$ over all vertices $v$ is the \emph{local
	clustering coefficient}, denoted by $\lcc$.
	\item Degree-wise clustering coefficients: We let $\dcc_d$ denote the average
clustering coefficient of degree-$d$ vertices. In addition to degree distribution, many graphs
are characterized by plotting the degree-wise clustering 
coefficients, i.e., $\dcc_d$, versus $d$.%
\end{asparaitem}
\medskip
In \Fig{example}, the transitivity is $\gcc = 3T/W=3\cdot 2/18=1/3$.
The clustering coefficients of the vertices are 0, 0, 1/3, 1/5, 1, 1, and 1 in order, and the local clustering coefficient is 0.5048. 
In general, the most direct method to compute these metrics is an exact computation
that finds all triangles. For degree-wise clustering coefficients, no other method was previously known.

Triangles and transitivity in directed graphs  have also been the subject of recent work (see e.g.~\cite{SePiDuKo13} and references therein). 
In a directed graph, edges are ordered pairs of vertices of the form $(i,j)$, indicating a link from node $i$ to node $j$.
When edges $(i,j)$ \emph{and} $(j,i)$ exist, we say there is a \emph{reciprocal edge} between $i$ and $j$. If just one edge exists, then we say it is a \emph{one-way edge}.
Considering all possible combinations of reciprocal and one-way edges leads to 6 different wedges and 7 different triangles.

\subsection{Our contributions \label{sec:results}}

In this paper, 
we discuss the simple yet powerful technique of \emph{wedge sampling} for counting
triangles. Wedge sampling is really an algorithmic template, since various algorithms can
be obtained from the same basic idea. It can be used to estimate all the different
triadic measures detailed above. The method
also provides precise probabilistic error bounds where the accuracy depends on the number of random samples.
The mathematical analysis of this method is a direct consequence of
standard Chernoff-Hoeffding bounds. If we want an estimate for $\gcc$ with
error $\eps \leq 0.1$ with probability 99.9\% (say), we need only 380
wedge samples. This estimate is \emph{independent of the size of the
  graph}, though the preprocessing required by our method is linear in
the number of edges (to obtain the degree distribution).

We discovered older independent work by Schank and Wagner that proposes the same wedge sampling idea for estimating the global and local clustering coefficients \cite{ScWa05-2}. Our work extends that in several directions, including several other uses for wedge sampling (such as directed triangle counting, random triangle sampling, degree-wise clustering coefficients) and much more extensive numerical results. This manuscript is an extended version of \cite{SePiKo13}; 
here we give more detailed experimental results and also show how wedge sampling applies to computing directed triangle counts.
A detailed list of contributions follows.

\begin{asparaitem}
\item \textbf{Computing the transitivity (global clustering
    coefficient) and local clustering coefficient:} We show how to use
  wedge sampling to approximate the global and local clustering
  coefficients: $\gcc$ and $\lcc$. We compare these methods to the
  state of the art, showing significant improvement for large-scale
  graphs.
  \item \textbf{Computing degree-wise clustering coefficients:} The idea of wedge sampling can be extended for more detailed measurements of graphs. We show how to calculate the degree-wise clustering coefficients, $\{\dcc_d\}$ for $d=1,\dots,d_{\max}$. The only other competing method that can compute $\{\dcc_d\}$ is an exhaustive enumeration.
  We compare with the basic fast enumeration algorithm given by~\cite{ChNi85,ScWa05} (which has been
  studied and reinvented by~\cite{Co09,SuVa11}).
  \item \textbf{Computing triangles per degree:} Wedge sampling can also be employed to sample random triangles, including the application of estimating the number of triangles containing one (or more) vertices of degree $d$, denoted $\{T_d\}$ for $d=1,\dots,d_{\max}$.
  \item \textbf{Counting directed triangles:} Few methods have been considered for counting triangles in directed graphs. It is especially complicated because there are 7 types of triangles to consider. Once again, the versatility of wedge sampling is used to develop a method for counting all types of directed triangles.
\end{asparaitem}

\subsection{Related Work}

There has been significant work on enumeration of all triangles~\cite{ChNi85,ScWa05,latapy08, BeFoNoPh11,ChCh11}.
Recent work by Cohen~\cite{Co09} and Suri and Vassilvitskii~\cite{SuVa11} give MapReduce implementations
of these algorithms. Arifuzzaman et al.~\cite{ArKhMa12} give a massively parallel algorithm for computing
clustering coefficients.  Enumeration algorithms however, can be very expensive, since  graphs even of moderate size (millions of vertices) can have
an extremely large number of triangles (see, e.g., \Tab{prop}). 
Eigenvalue/trace based methods have been used by Tsourakakis~\cite{Ts08} and Avron~\cite{Av10} to compute
estimates of the total and per-degree number of triangles. However, computing eigenvalues (even just a few of them) is a compute-intensive task and  quickly becomes intractable on large graphs. In our experiment, even computing the largest eigenvalue was multiple orders of magnitude slower than full enumeration.

Most relevant to our work are sampling mechanisms.
Tsourakakis et al.~\cite{TsDrMi09} started the use of sparsification methods, the most important of which
is Doulion~\cite{TsKaMiFa09}.  This method sparsifies the graph by keeping each edge with probability $p$; counts  the triangles  in the sparsified graph; and multiplies this count by $p^{-3}$ to predict the  number of triangles in the original graph.  
Various theoretical analyses of this algorithm (and its variants) have been
proposed~\cite{KoMiPeTs10,TsKoMi11,PaTs12}. 
One of the main benefits of Doulion is that it reduces large graphs to smaller ones that can be loaded into memory. 
However, the Doulion estimate can suffer from high variance~\cite{YoKi11}.
Alternative sampling mechanisms have been proposed for streaming and semi-streaming algorithms \cite{BaKuSi02, JoGh05, BeBoCaGi08,BuFrLeMaSo06}.
Yet, all these fast sampling methods only estimate the number of triangles
and give no information about other triadic measures.

In subsequent work by the authors of this paper, a Hadoop implementation of these techniques is given in ~\cite{KoPiPlSe13}, 
and a streaming version of the wedge sampling method is presented in~\cite{JhSePi13}.

\section{Overview of  wedge sampling} \label{sec:wedge}

We present the general method of wedge sampling for estimating clustering coefficients. 
In later sections, we instantiate this for different algorithms.
Before we begin, we summarize the notation presented in \Sec{intro} in \Tab{notation}.

\begin{table}[t]
  \caption{Graph notation and clustering coefficients for undirected graphs}
  \label{tab:notation}
  \begin{tabularx}{\linewidth}{lX|lX}
    \toprule
    {$n$} &  number of vertices &
    {$n_d$} & number of vertices of degree $d$\\
    {$m$} & number of edges &
    {$d_v$} & degree of vertex $v$\\
    {$V_d$} & set of degree-$d$ vertices &
    {$W$} &  total number of wedges\\
    {$W_v$} &  number of wedges centered at vertex $v$&
    {$T$} & total number of triangles \\
    {$T_v$} & number of triangles incident to vertex $v$& 
    {$T_d$} & \raggedright number of triangles incident to degree-$d$ vertices
  \end{tabularx}
  \begin{tabularx}{\linewidth}{@{}lX@{}}
    \toprule
    $\gcc = 3T/W$ & transitivity\\
    $C_v = T_v/W_v$ & clustering coefficient of vertex $v$\\
    $ \lcc = n^{-1} \sum_v C_v$ & local clustering coefficient\\
    $ \dcc_d = n_d^{-1} \sum_{v \in V_d} C_v$ & degree-wise clustering coefficient\\
    \bottomrule
  \end{tabularx}
\end{table}

We say a wedge is \emph{closed} if it is part of a
triangle; otherwise, the wedge is \emph{open}. Thus,
in \Fig{example}, 
\circled{5}-\circled{4}-\circled{6} is an open wedge, while
\circled{3}-\circled{4}-\circled{5} is a closed wedge.
The middle vertex of a wedge is called its \emph{center}, i.e.,
wedges \circled{5}-\circled{4}-\circled{6} and
\circled{3}-\circled{4}-\circled{5} are centered at
\circled{4}.

Wedge sampling is best understood in terms of the following thought experiment.
Fix some distribution over wedges and
let $w$ be a random wedge. 
Let $X$ be the indicator random variable that is $1$
if $w$ is closed and $0$ otherwise.  Denote $\mu = \EX[X]$.

Suppose we wish to estimate $\mu$. We simply generate $k$
independent random wedges $w_1, w_2, \ldots, w_k$, with associated
random variables $X_1, X_2, \ldots, X_k$. Define
$\bar{X} = \frac{1}{k} \sum_{i \leq k} X_i$ as our estimate. The
Chernoff-Hoeffding bounds give guarantees on $\bar X$, as follows.
\begin{theorem}[Hoeffding \cite{Ho63}]
  \label{thm:Hoeffding}
  Let $X_1, X_2, \dots, X_k$ be independent random variables with $0
  \leq X_i \leq 1$ for all $i=1,\dots,k$.  Define $\bar X =
  \frac{1}{k} \sum_{i=1}^k X_i$. Let $\mu = \EX[\bar X]$. 
  Then for $\eps \in (0,1)$, we have
  \begin{displaymath}
    \Prob{ |\bar X - \mu | \geq \eps } \leq 2 \exp(-2 k \eps^2).
  \end{displaymath}
\end{theorem}
Hence, if we set $k = \cei{0.5 \eps^{-2}\ln(2/\delta)}$, then
$\Pr[|\bar{X} - \mu| > \eps] < \delta$. In other words,
for $k$ samples,
with confidence at least $1-\delta$, 
the error in our estimate is at most $\eps$.

\Fig{level_curves} shows the number of samples needed
as the error rate varies. We show three curves for
different confidence levels. Increasing the confidence has minimal
impact on the number of samples. The number of samples is fairly low
for error rates of 0.1 or 0.01, but it increases with the inverse
square of $\eps$. 

\begin{figure}[htp]
  \centering
  \includegraphics{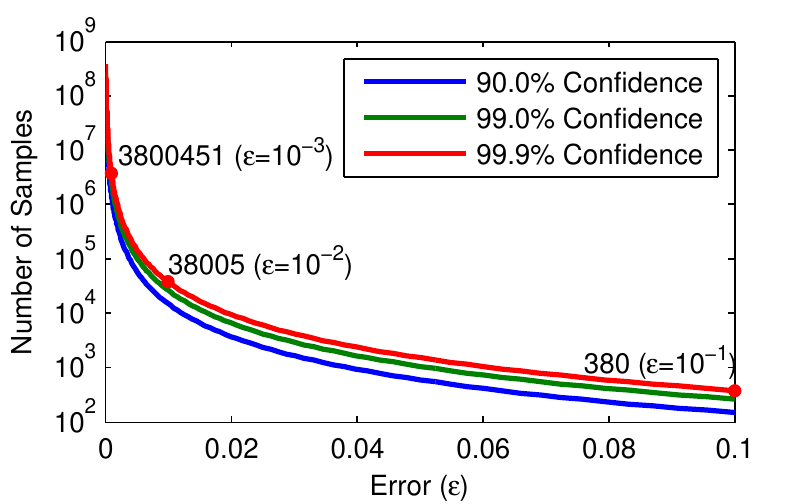}
  \caption{The number of samples needed for different error rates and
    different levels of confidence. A few data points at 99.9\%
    confidence are highlighted.}
  \label{fig:level_curves}
\end{figure}

\section{Computing the transitivity and the number of triangles} \label{sec:gcc}

We use the wedge sampling scheme to estimate the transitivity, $\gcc$. 
Consider the uniform distribution on wedges. We can interpret
$\EX[X]$ as the probability that a uniform random wedge is closed
or, alternately, the fraction of closed wedges.

To generate a uniform random wedge, note that the number of wedges
centered at vertex $v$ is $W_v = {d_v \choose 2}$ and $W = \sum_v
W_v$.  We set $p_v = W_v/W$ to get a distribution over the vertices.
Note that the probability of picking $v$ is proportional to the number
of wedges centered at $v$.
A uniform random wedge centered at $v$ can be generated by choosing two
random neighbors of $v$ (without replacement).

\begin{claim} \label{clm:random} 
  Suppose we choose vertex $v$ with probability $p_v$
  and take a uniform random pair of neighbors of $v$. 
  This generates a uniform random wedge.
\end{claim}

\begin{proof} Consider some wedge $w$ centered at vertex $v$. 
The probability that $v$ is selected is $p_v=W_v/W$.
The random pair has probability of $1/{d_v \choose 2} = 1/W_v$.
Hence, the net probability of $w$ is $1/W$.
\qed
\end{proof}

\Alg{C} shows the randomized algorithm \algkappa{} for estimating $\gcc$ in a graph $G$. Observe that the first step assumes that the degree of each vertex is already computed. Sampling with replacements means that we sample from the original distribution repeatedly and repeats are allowed. Sampling without replacement means that we cannot pick the same item twice.

\begin{algorithm}
\caption{ Computing the transitivity (\algkappa{}) }
\label{alg:C}
\begin{algorithmic}[1]
\STATE	Select $k$ random vertices (with replacement) according to the
probability distribution defined by $\{p_v\}$ where $p_v = W_v/W$.
\STATE	 For each selected vertex $v$, choose two
neighbors of $v$ (uniformly at random without replacement) to generate a random wedge. The set of all wedges comprises the sample set.
\STATE	Output the fraction of closed wedges in the sample set as estimate of $\gcc$.
\end{algorithmic}
\end{algorithm}

Combining the bound of \Thm{Hoeffding} with \Clm{random}, we get the following theorem. Note that the number
of samples required is \emph{independent of the graph size}, but
computing $p_v$ does depend on the number of edges, $m$. 

\begin{theorem} \label{thm:kappa} 
Set $k = \cei{0.5 \eps^{-2}\ln(2/\delta)}$. The algorithm \algkappa{}
outputs an estimate 
$\bar X$ for the transitivity $\gcc$
such that $|\bar{X} - \gcc| < \eps$
with probability greater than $(1-\delta)$.
\end{theorem}

To get an estimate on $T$, the number of triangles, we output $\bar{X}
\cdot W/3$, which is guaranteed to be within $\pm \eps W/3$ of
$T$ with probability greater than $1-\delta$.

\newcommand\cta[1]{\multicolumn{1}{|c|}{#1}}
\newcommand\ct[1]{\multicolumn{1}{c|}{#1}}
\begin{table}[t]
 \caption{Properties of the graphs and runtimes for enumeration.}
\label{tab:prop}
  \centering\small
  \begin{tabular}{|>{\tt}r@{\,}|*{7}{r@{\,}|}}
    \hline
            \cta{Graph} & \ct{Nodes} & \ct{Edges} & \ct{Wedges} & \ct{Triangles} & \ct{Transitivity} & \ct{Local CC} & \ct{Enum.} \\
            \cta{Name} &\ct{$n$}&\ct{$m$}&\ct{$W$}&\ct{$T$}&\ct {$\gcc$} &\ct{$\lcc$}& Time (secs) \\ %
    \hline
             amazon0312 &   401K &  2350K &    69M &  3686K &	0.160 & 0.421 & 0.261	\\
             amazon0505 &   410K &  2439K &    73M &  3951K &  	0.162 & 0.427& 0.269 \\
             amazon0601 &   403K &  2443K &    72M &  3987K &  0.166 & 0.430& 0.268	 \\
             as-skitter &  1696K & 11095K & 16022M & 28770K & 0.005	& 0.296 & 90.897	\\
                        cit-Patents &  3775K & 16519K &   336M &  7515K & 0.067 & 0.092& 3.764 \\
            	             roadNet-CA &  1965K &  2767K &     6M &   121K & 0.060& 0.055 & 0.095 \\
          	           web-BerkStan &   685K &  6649K & 27983M & 64691K & 0.007& 0.634  & 54.777 \\
             web-Google &   876K &  4322K &   727M & 13392K & 0.055  & 0.624& 0.894	\\
           web-Stanford &   282K &  1993K &  3944M & 11329K & 0.009 & 0.629& 6.987  \\
              wiki-Talk &  2394K &  4660K & 12594M &  9204K &  0.002 & 0.201&  20.572	\\
              youtube & 1158K&  2990K& 1474M& 3057K& 0.006& 0.128 & 2.740 \\
              flickr & 1861K& 15555K  &14670M & 548659K &   0.112 &  0.375 & 567.160 \\
              livejournal & 5284K & 48710K& 7519M& 310877K& 0.124 & 0.345& 102.142 \\
              \hline
 \end{tabular}
\end{table}

We  implemented our algorithms in {C} and ran our experiments on a
computer equipped with a 2.3GHz Intel core i7 processor with 4~cores
and  256KB  L2 cache (per core), 8MB L3 cache, and an 8GB memory.  
We performed our experiments on 13 graphs  from
SNAP~\cite{Snap} and per private communication
with the authors of~\cite{MiMaGu07}.
In all cases, directionality is ignored, and repeated and  self-edges are omitted. 
The properties of these matrices are presented in \Tab{prop}.
The last column reports the times for the enumeration algorithm. This is based on
the principles of \cite{ChNi85, ScWa05, Co09, SuVa11}: each
edge is assigned to the vertex with a smaller degree  (using the
vertex numbering as a tie-breaker), and then vertices only check
wedges formed by edges assigned to them for closure. 

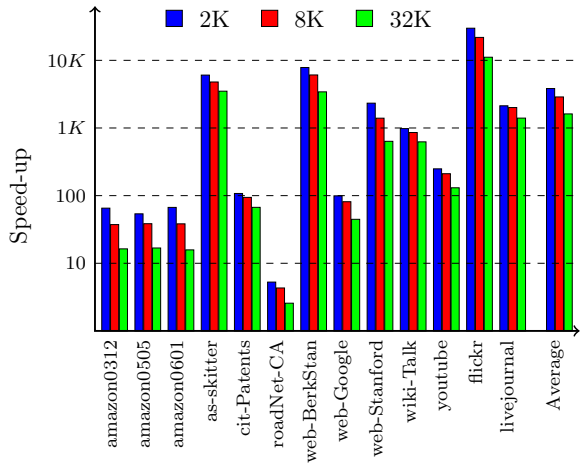
\begin{figure}
  \centering
\scalebox{0.9}{
\begin{tikzpicture}\draw [fill=\typea,thin] (0.100,0) rectangle (0.220,1.815) ;
\draw [fill=\typeb,thin] (0.230,0) node[below]{\rotatebox[origin=t]{90}{\scriptsize amazon0312}} rectangle (0.350,1.572) ;
\draw [fill=\typec,thin] (0.360,0) rectangle (0.480,1.213) ;
\draw [fill=\typea,thin] (0.590,0) rectangle (0.710,1.731) ;
\draw [fill=\typeb,thin] (0.720,0) node[below]{\rotatebox[origin=t]{90}{\scriptsize amazon0505}} rectangle (0.840,1.585) ;
\draw [fill=\typec,thin] (0.850,0) rectangle (0.970,1.226) ;
\draw [fill=\typea,thin] (1.080,0) rectangle (1.200,1.826) ;
\draw [fill=\typeb,thin] (1.210,0) node[below]{\rotatebox[origin=t]{90}{\scriptsize amazon0601}} rectangle (1.330,1.583) ;
\draw [fill=\typec,thin] (1.340,0) rectangle (1.460,1.198) ;
\draw [fill=\typea,thin] (1.570,0) rectangle (1.690,3.782) ;
\draw [fill=\typeb,thin] (1.700,0) node[below]{\rotatebox[origin=t]{90}{\scriptsize as-skitter}} rectangle (1.820,3.680) ;
\draw [fill=\typec,thin] (1.830,0) rectangle (1.950,3.544) ;
\draw [fill=\typea,thin] (2.060,0) rectangle (2.180,2.032) ;
\draw [fill=\typeb,thin] (2.190,0) node[below]{\rotatebox[origin=t]{90}{\scriptsize cit-Patents}} rectangle (2.310,1.974) ;
\draw [fill=\typec,thin] (2.320,0) rectangle (2.440,1.827) ;
\draw [fill=\typea,thin] (2.550,0) rectangle (2.670,0.722) ;
\draw [fill=\typeb,thin] (2.680,0) node[below]{\rotatebox[origin=t]{90}{\scriptsize roadNet-CA}} rectangle (2.800,0.635) ;
\draw [fill=\typec,thin] (2.810,0) rectangle (2.930,0.410) ;
\draw [fill=\typea,thin] (3.040,0) rectangle (3.160,3.894) ;
\draw [fill=\typeb,thin] (3.170,0) node[below]{\rotatebox[origin=t]{90}{\scriptsize web-BerkStan}} rectangle (3.290,3.784) ;
\draw [fill=\typec,thin] (3.300,0) rectangle (3.420,3.534) ;
\draw [fill=\typea,thin] (3.530,0) rectangle (3.650,1.997) ;
\draw [fill=\typeb,thin] (3.660,0) node[below]{\rotatebox[origin=t]{90}{\scriptsize web-Google}} rectangle (3.780,1.910) ;
\draw [fill=\typec,thin] (3.790,0) rectangle (3.910,1.650) ;
\draw [fill=\typea,thin] (4.020,0) rectangle (4.140,3.367) ;
\draw [fill=\typeb,thin] (4.150,0) node[below]{\rotatebox[origin=t]{90}{\scriptsize web-Stanford}} rectangle (4.270,3.145) ;
\draw [fill=\typec,thin] (4.280,0) rectangle (4.400,2.803) ;
\draw [fill=\typea,thin] (4.510,0) rectangle (4.630,2.991) ;
\draw [fill=\typeb,thin] (4.640,0) node[below]{\rotatebox[origin=t]{90}{\scriptsize wiki-Talk}} rectangle (4.760,2.933) ;
\draw [fill=\typec,thin] (4.770,0) rectangle (4.890,2.795) ;
\draw [fill=\typea,thin] (5.000,0) rectangle (5.120,2.396) ;
\draw [fill=\typeb,thin] (5.130,0) node[below]{\rotatebox[origin=t]{90}{\scriptsize youtube}} rectangle (5.250,2.324) ;
\draw [fill=\typec,thin] (5.260,0) rectangle (5.380,2.116) ;
\draw [fill=\typea,thin] (5.490,0) rectangle (5.610,4.475) ;
\draw [fill=\typeb,thin] (5.620,0) node[below]{\rotatebox[origin=t]{90}{\scriptsize flickr}} rectangle (5.740,4.339) ;
\draw [fill=\typec,thin] (5.750,0) rectangle (5.870,4.046) ;
\draw [fill=\typea,thin] (5.980,0) rectangle (6.100,3.328) ;
\draw [fill=\typeb,thin] (6.110,0) node[below]{\rotatebox[origin=t]{90}{\scriptsize livejournal}} rectangle (6.230,3.302) ;
\draw [fill=\typec,thin] (6.240,0) rectangle (6.360,3.146) ;
\draw [fill=\typea,thin] (6.670,0) rectangle (6.790,3.583) ;
\draw [fill=\typeb,thin] (6.800,0) node[below]{\rotatebox[origin=t]{90}{\scriptsize Average}} rectangle (6.920,3.459) ;
\draw [fill=\typec,thin] (6.930,0) rectangle (7.050,3.208) ;
\draw [<->, thick] (7.16,0) -- (0,0)-- (0,0.00) node[left] at (-0.8, 2.00) {\rotatebox{90}{\small Speed-up}} -- (0, 4.80) ;
\draw [dashed] (0, 1.000) node [left]{\scriptsize $10$} -- (7.080,1.000);
\draw [dashed] (0, 2.000) node [left]{\scriptsize$100$} -- (7.080,2.000);
\draw [dashed] (0, 3.000) node [left]{\scriptsize$1K$} -- (7.080,3.000);
\draw [dashed] (0, 4.000) node [left]{\scriptsize$10K$} -- (7.080,4.000);
\draw [fill=\typea,thin] (1.05,4.50) node[right] at(1.41,4.62) {\small 2K} rectangle (1.29,4.74)  ;
\draw [fill=\typeb,thin] (2.46,4.50) node[right]at(2.81,4.62) {\small 8K} rectangle (2.70,4.74)  ;
\draw [fill=\typec,thin] (3.86,4.50) node[right]at(4.22,4.62) {\small 32K} rectangle (4.10,4.74)  ;
\end{tikzpicture}
}
\caption{ Speed-up over enumeration (in log-scale) for transitivity computation with increasing numbers of wedge samples.}
   \label{fig:gcc-time}
\end{figure}
As  seen in \Fig{gcc-time},  wedge sampling is orders of magnitude faster than the enumeration algorithm. 
The timing results show tremendous savings; for instance,   wedge
sampling only takes 0.026 seconds on \texttt{as-skitter} while full
enumeration takes 90 seconds.  

\begin{figure} 
\centering 
 \vspace*{-5ex}
\scalebox{1.0}{
\begin{tikzpicture}\draw [fill=\typea,thin] (0.000,0) rectangle (0.120,1.800) ;
\draw [fill=\typeb,thin] (0.130,0) node[below]{\rotatebox[origin=t]{90}{\scriptsize amazon0312}} rectangle (0.250,0.600) ;
\draw [fill=\typec,thin] (0.260,0) rectangle (0.380,0.060) ;
\draw [fill=\typea,thin] (0.490,0) rectangle (0.610,2.400) ;
\draw [fill=\typeb,thin] (0.620,0) node[below]{\rotatebox[origin=t]{90}{\scriptsize amazon0505}} rectangle (0.740,1.800) ;
\draw [fill=\typec,thin] (0.750,0) rectangle (0.870,0.600) ;
\draw [fill=\typea,thin] (0.980,0) rectangle (1.100,3.000) ;
\draw [fill=\typeb,thin] (1.110,0) node[below]{\rotatebox[origin=t]{90}{\scriptsize amazon0601}} rectangle (1.230,1.200) ;
\draw [fill=\typec,thin] (1.240,0) rectangle (1.360,0.600) ;
\draw [fill=\typea,thin] (1.470,0) rectangle (1.590,0.600) ;
\draw [fill=\typeb,thin] (1.600,0) node[below]{\rotatebox[origin=t]{90}{\scriptsize as-skitter}} rectangle (1.720,0.600) ;
\draw [fill=\typec,thin] (1.730,0) rectangle (1.850,0.600) ;
\draw [fill=\typea,thin] (1.960,0) rectangle (2.080,1.800) ;
\draw [fill=\typeb,thin] (2.090,0) node[below]{\rotatebox[origin=t]{90}{\scriptsize cit-Patents}} rectangle (2.210,0.060) ;
\draw [fill=\typec,thin] (2.220,0) rectangle (2.340,0.600) ;
\draw [fill=\typea,thin] (2.450,0) rectangle (2.570,0.600) ;
\draw [fill=\typeb,thin] (2.580,0) node[below]{\rotatebox[origin=t]{90}{\scriptsize roadNet-CA}} rectangle (2.700,1.200) ;
\draw [fill=\typec,thin] (2.710,0) rectangle (2.830,1.200) ;
\draw [fill=\typea,thin] (2.940,0) rectangle (3.060,1.200) ;
\draw [fill=\typeb,thin] (3.070,0) node[below]{\rotatebox[origin=t]{90}{\scriptsize web-BerkStan}} rectangle (3.190,0.600) ;
\draw [fill=\typec,thin] (3.200,0) rectangle (3.320,0.060) ;
\draw [fill=\typea,thin] (3.430,0) rectangle (3.550,0.060) ;
\draw [fill=\typeb,thin] (3.560,0) node[below]{\rotatebox[origin=t]{90}{\scriptsize web-Google}} rectangle (3.680,0.600) ;
\draw [fill=\typec,thin] (3.690,0) rectangle (3.810,0.600) ;
\draw [fill=\typea,thin] (3.920,0) rectangle (4.040,2.400) ;
\draw [fill=\typeb,thin] (4.050,0) node[below]{\rotatebox[origin=t]{90}{\scriptsize web-Stanford}} rectangle (4.170,0.600) ;
\draw [fill=\typec,thin] (4.180,0) rectangle (4.300,0.600) ;
\draw [fill=\typea,thin] (4.410,0) rectangle (4.530,1.200) ;
\draw [fill=\typeb,thin] (4.540,0) node[below]{\rotatebox[origin=t]{90}{\scriptsize wiki-Talk}} rectangle (4.660,0.600) ;
\draw [fill=\typec,thin] (4.670,0) rectangle (4.790,0.600) ;
\draw [fill=\typea,thin] (4.900,0) rectangle (5.020,0.600) ;
\draw [fill=\typeb,thin] (5.030,0) node[below]{\rotatebox[origin=t]{90}{\scriptsize youtube}} rectangle (5.150,0.060) ;
\draw [fill=\typec,thin] (5.160,0) rectangle (5.280,0.060) ;
\draw [fill=\typea,thin] (5.390,0) rectangle (5.510,1.200) ;
\draw [fill=\typeb,thin] (5.520,0) node[below]{\rotatebox[origin=t]{90}{\scriptsize flickr}} rectangle (5.640,0.600) ;
\draw [fill=\typec,thin] (5.650,0) rectangle (5.770,0.060) ;
\draw [fill=\typea,thin] (5.880,0) rectangle (6.000,1.800) ;
\draw [fill=\typeb,thin] (6.010,0) node[below]{\rotatebox[origin=t]{90}{\scriptsize livejournal}} rectangle (6.130,1.200) ;
\draw [fill=\typec,thin] (6.140,0) rectangle (6.260,0.060) ;
\draw [fill=\typea,thin] (6.570,0) rectangle (6.690,1.431) ;
\draw [fill=\typeb,thin] (6.700,0) node[below]{\rotatebox[origin=t]{90}{\scriptsize Average}} rectangle (6.820,0.738) ;
\draw [fill=\typec,thin] (6.830,0) rectangle (6.950,0.323) ;
\draw [<->, thick] (7.16,0) -- (0,0)-- (0,1.80) node[left]{\vspace{10ex}\hspace{-15ex}\rotatebox{90}{\small Absolute error}} -- (0, 3.60) ;
\draw [dashed] (0, 0.900) node [left]{\scriptsize $0.0015$} -- (6.980,0.900);
\draw [dashed] (0, 1.800) node [left]{\scriptsize$0.003$} -- (6.980,1.800);
\draw [dashed] (0, 2.700) node [left]{\scriptsize$0.0045$} -- (6.980,2.700);
\draw [fill=\typea,thin] (1.05,3.60) node[right] at(1.41,3.72) {\small 2K} rectangle (1.29,3.84)  ;
\draw [fill=\typeb,thin] (2.46,3.60) node[right]at(2.81,3.72) {\small 8K} rectangle (2.70,3.84)  ;
\draw [fill=\typec,thin] (3.86,3.60) node[right]at(4.22,3.72) {\small 32K} rectangle (4.10,3.84)  ;
\end{tikzpicture}
}
\caption{ Absolute error in transitivity for increasing numbers of wedge samples.}
\label{fig:gcc-error}
\end{figure}
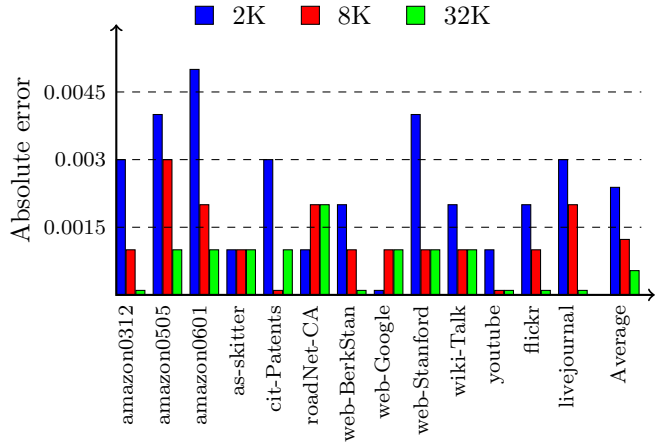

\Fig{gcc-error} shows the accuracy of  the wedge sampling algorithm. 
At 99.9\% confidence ($\delta=0.001$), the upper bound on the error we expect for 2K, 8K,
and 32K samples is .043, .022, and .011, respectively. Most of the
errors are much less than the bounds would suggest. For instance, the
maximum error for 2K samples is .007, much less than that 0.43 given
by the upper bound.

We show the rate of convergence in \Fig{amazon0505_cc} for the graph {\tt amazon0505} as the number of samples increases. The dashed
line shows the error bars at 99.9\% confidence.  Note that the convergence is fast and the error bars are conservative in this instance.
\begin{figure}[tp]
  \centering
  \includegraphics{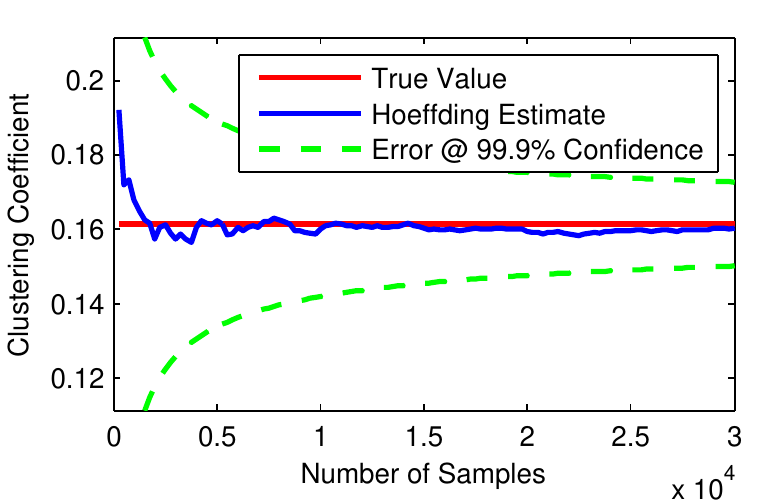}
  \caption{Convergence of transitivity estimate as the
    number of samples increases for  \texttt{amazon0505}.}
  \label{fig:amazon0505_cc}
\end{figure}

\section{Computing the local clustering coefficient} \label{sec:lcc}

We now demonstrate how a small change to the underlying distribution on wedges
allows us to compute the clustering coefficient,
$\lcc$. \Alg{lcc} shows
the procedure \alglcc{}.
The only difference between  \alglcc{} and \algkappa{} is in picking
random vertex centers. Vertices are picked uniformly instead of from the distribution $\{p_v\}$. 
\medskip
\begin{algorithm}
\caption{ Computing the local clustering coefficient (\alglcc{}) \label{alg:lcc}}
\begin{algorithmic}[1]
\STATE Pick $k$ uniform random vertices (with replacement).
\STATE For each selected vertex $v$, choose two
neighbors of $v$ (uniformly at random without replacement) to generate a random wedge. The set of all wedges comprises the sample set.
\STATE Output the fraction of closed wedges in the sample set as an estimate for $\lcc$.
\end{algorithmic}
\end{algorithm}

\begin{theorem} \label{thm:lcc} 
Set $k = \cei{0.5 \eps^{-2}\ln(2/\delta)}$. The algorithm \alglcc{}
outputs an estimate 
$\bar X$ for the clustering coefficient $\lcc$ such that 
$|\bar{X} - \lcc| < \eps$
with probability greater than $(1-\delta)$.
\end{theorem}

\begin{proof} 
Let us consider a single sampled wedge $w$, and let $X(w)$
be the indicator random variable for the wedge being closed. Let $\cV$
be the uniform distribution on edges. For any vertex $v$, let $\cN_v$
be the uniform distribution on pairs of neighbors of $v$.
Observe that 
$$ {\EX[X]} = \Pr_{v \sim \cV}[\Pr_{(u,u') \sim \cN_v}[\textrm{wedge $\{(u,v),(u',v)\}$ is closed}]] $$
We will show that this is exactly $\lcc$.
\begin{eqnarray*}
\lcc & = & n^{-1} \sum_v C_v = \EX_{v \sim \cV} [C_v] \\
& = & \EX_{v \sim \cV} [\textrm{frac. of closed wedges centered at $v$}] \\
& = & \EX_{v \sim \cV} [\Pr_{(u,u') \sim \cN_v}[\textrm{$\{(u,v),(u',v)\}$ is closed}]] \\
& = & \Pr_{v \sim \cV}[\Pr_{(u,u') \sim \cN_v}[\textrm{$\{(u,v),(u',v)\}$ is closed}]]
= \EX[X]
\end{eqnarray*}
For a single sample, the probability that the wedge is closed is exactly $\lcc$.
The bound of \Thm{Hoeffding} completes the proof.
\qed
\end{proof}

\Fig{lcc} presents  the results of our experiments for computing the clustering coefficients. 
Experimental setup and the notation are  the same as before. 
The results again show that wedge sampling provides accurate
estimations with tremendous improvements in runtime. For instance, for 8K samples, the average speedup is 10K. 

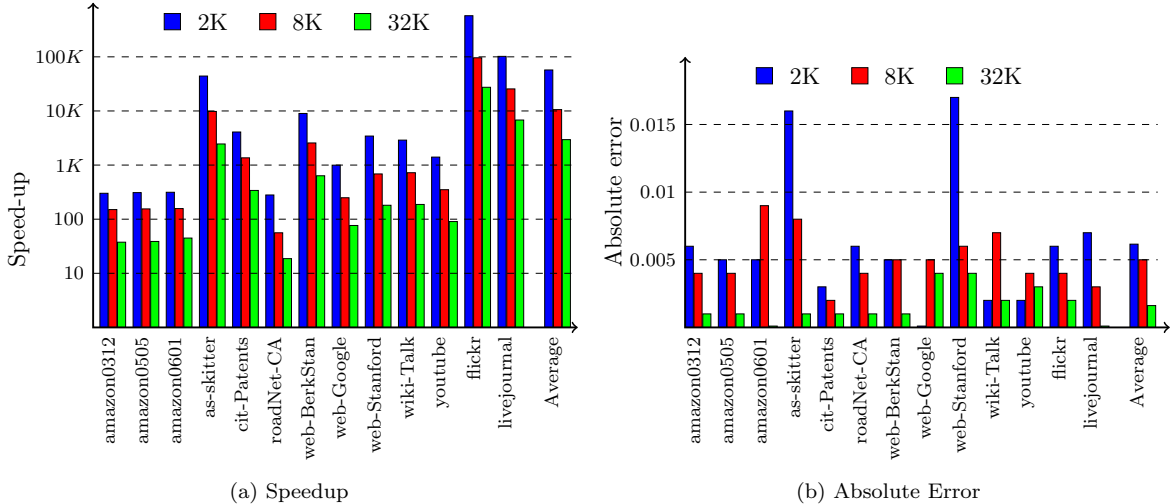
\begin{figure}[ht]
  \centering
\subfloat[Speedup]{\label{fig:lcc-times}\scalebox{0.9}{
\begin{tikzpicture}\draw [fill=\typea,thin] (0.100,0) rectangle (0.220,1.983) ;
\draw [fill=\typeb,thin] (0.230,0) node[below]{\rotatebox[origin=t]{90}{\scriptsize amazon0312}} rectangle (0.350,1.742) ;
\draw [fill=\typec,thin] (0.360,0) rectangle (0.480,1.260) ;
\draw [fill=\typea,thin] (0.590,0) rectangle (0.710,1.993) ;
\draw [fill=\typeb,thin] (0.720,0) node[below]{\rotatebox[origin=t]{90}{\scriptsize amazon0505}} rectangle (0.840,1.752) ;
\draw [fill=\typec,thin] (0.850,0) rectangle (0.970,1.271) ;
\draw [fill=\typea,thin] (1.080,0) rectangle (1.200,1.998) ;
\draw [fill=\typeb,thin] (1.210,0) node[below]{\rotatebox[origin=t]{90}{\scriptsize amazon0601}} rectangle (1.330,1.757) ;
\draw [fill=\typec,thin] (1.340,0) rectangle (1.460,1.321) ;
\draw [fill=\typea,thin] (1.570,0) rectangle (1.690,3.716) ;
\draw [fill=\typeb,thin] (1.700,0) node[below]{\rotatebox[origin=t]{90}{\scriptsize as-skitter}} rectangle (1.820,3.193) ;
\draw [fill=\typec,thin] (1.830,0) rectangle (1.950,2.712) ;
\draw [fill=\typea,thin] (2.060,0) rectangle (2.180,2.889) ;
\draw [fill=\typeb,thin] (2.190,0) node[below]{\rotatebox[origin=t]{90}{\scriptsize cit-Patents}} rectangle (2.310,2.507) ;
\draw [fill=\typec,thin] (2.320,0) rectangle (2.440,2.025) ;
\draw [fill=\typea,thin] (2.550,0) rectangle (2.670,1.958) ;
\draw [fill=\typeb,thin] (2.680,0) node[below]{\rotatebox[origin=t]{90}{\scriptsize roadNet-CA}} rectangle (2.800,1.399) ;
\draw [fill=\typec,thin] (2.810,0) rectangle (2.930,1.017) ;
\draw [fill=\typea,thin] (3.040,0) rectangle (3.160,3.163) ;
\draw [fill=\typeb,thin] (3.170,0) node[below]{\rotatebox[origin=t]{90}{\scriptsize web-BerkStan}} rectangle (3.290,2.727) ;
\draw [fill=\typec,thin] (3.300,0) rectangle (3.420,2.242) ;
\draw [fill=\typea,thin] (3.530,0) rectangle (3.650,2.399) ;
\draw [fill=\typeb,thin] (3.660,0) node[below]{\rotatebox[origin=t]{90}{\scriptsize web-Google}} rectangle (3.780,1.917) ;
\draw [fill=\typec,thin] (3.790,0) rectangle (3.910,1.507) ;
\draw [fill=\typea,thin] (4.020,0) rectangle (4.140,2.829) ;
\draw [fill=\typeb,thin] (4.150,0) node[below]{\rotatebox[origin=t]{90}{\scriptsize web-Stanford}} rectangle (4.270,2.269) ;
\draw [fill=\typec,thin] (4.280,0) rectangle (4.400,1.806) ;
\draw [fill=\typea,thin] (4.510,0) rectangle (4.630,2.769) ;
\draw [fill=\typeb,thin] (4.640,0) node[below]{\rotatebox[origin=t]{90}{\scriptsize wiki-Talk}} rectangle (4.760,2.287) ;
\draw [fill=\typec,thin] (4.770,0) rectangle (4.890,1.818) ;
\draw [fill=\typea,thin] (5.000,0) rectangle (5.120,2.518) ;
\draw [fill=\typeb,thin] (5.130,0) node[below]{\rotatebox[origin=t]{90}{\scriptsize youtube}} rectangle (5.250,2.036) ;
\draw [fill=\typec,thin] (5.260,0) rectangle (5.380,1.565) ;
\draw [fill=\typea,thin] (5.490,0) rectangle (5.610,4.608) ;
\draw [fill=\typeb,thin] (5.620,0) node[below]{\rotatebox[origin=t]{90}{\scriptsize flickr}} rectangle (5.740,3.986) ;
\draw [fill=\typec,thin] (5.750,0) rectangle (5.870,3.550) ;
\draw [fill=\typea,thin] (5.980,0) rectangle (6.100,4.007) ;
\draw [fill=\typeb,thin] (6.110,0) node[below]{\rotatebox[origin=t]{90}{\scriptsize livejournal}} rectangle (6.230,3.526) ;
\draw [fill=\typec,thin] (6.240,0) rectangle (6.360,3.066) ;
\draw [fill=\typea,thin] (6.670,0) rectangle (6.790,3.806) ;
\draw [fill=\typeb,thin] (6.800,0) node[below]{\rotatebox[origin=t]{90}{\scriptsize Average}} rectangle (6.920,3.220) ;
\draw [fill=\typec,thin] (6.930,0) rectangle (7.050,2.776) ;
\draw [<->, thick] (7.16,0) -- (0,0)-- (0,0.00) node[left] at (-0.8, 1.60) {\rotatebox{90}{\small Speed-up}} -- (0, 4.80) ;
\draw [dashed] (0, 0.800) node [left]{\scriptsize $10$} -- (7.080,0.800);
\draw [dashed] (0, 1.600) node [left]{\scriptsize$100$} -- (7.080,1.600);
\draw [dashed] (0, 2.400) node [left]{\scriptsize$1K$} -- (7.080,2.400);
\draw [dashed] (0, 3.200) node [left]{\scriptsize$10K$} -- (7.080,3.200);
\draw [dashed] (0, 4.000) node [left]{\scriptsize$100K$} -- (7.080,4.000);
\draw [fill=\typea,thin] (1.05,4.40) node[right] at(1.41,4.52) {\small 2K} rectangle (1.29,4.64)  ;
\draw [fill=\typeb,thin] (2.46,4.40) node[right]at(2.81,4.52) {\small 8K} rectangle (2.70,4.64)  ;
\draw [fill=\typec,thin] (3.86,4.40) node[right]at(4.22,4.52) {\small 32K} rectangle (4.10,4.64)  ;
\end{tikzpicture}
}}
\subfloat[Absolute Error]{\scalebox{0.9}{ 
\begin{tikzpicture}\draw [fill=\typea,thin] (0.000,0) rectangle (0.120,1.200) ;
\draw [fill=\typeb,thin] (0.130,0) node[below]{\rotatebox[origin=t]{90}{\scriptsize amazon0312}} rectangle (0.250,0.800) ;
\draw [fill=\typec,thin] (0.260,0) rectangle (0.380,0.200) ;
\draw [fill=\typea,thin] (0.490,0) rectangle (0.610,1.000) ;
\draw [fill=\typeb,thin] (0.620,0) node[below]{\rotatebox[origin=t]{90}{\scriptsize amazon0505}} rectangle (0.740,0.800) ;
\draw [fill=\typec,thin] (0.750,0) rectangle (0.870,0.200) ;
\draw [fill=\typea,thin] (0.980,0) rectangle (1.100,1.000) ;
\draw [fill=\typeb,thin] (1.110,0) node[below]{\rotatebox[origin=t]{90}{\scriptsize amazon0601}} rectangle (1.230,1.800) ;
\draw [fill=\typec,thin] (1.240,0) rectangle (1.360,0.020) ;
\draw [fill=\typea,thin] (1.470,0) rectangle (1.590,3.200) ;
\draw [fill=\typeb,thin] (1.600,0) node[below]{\rotatebox[origin=t]{90}{\scriptsize as-skitter}} rectangle (1.720,1.600) ;
\draw [fill=\typec,thin] (1.730,0) rectangle (1.850,0.200) ;
\draw [fill=\typea,thin] (1.960,0) rectangle (2.080,0.600) ;
\draw [fill=\typeb,thin] (2.090,0) node[below]{\rotatebox[origin=t]{90}{\scriptsize cit-Patents}} rectangle (2.210,0.400) ;
\draw [fill=\typec,thin] (2.220,0) rectangle (2.340,0.200) ;
\draw [fill=\typea,thin] (2.450,0) rectangle (2.570,1.200) ;
\draw [fill=\typeb,thin] (2.580,0) node[below]{\rotatebox[origin=t]{90}{\scriptsize roadNet-CA}} rectangle (2.700,0.800) ;
\draw [fill=\typec,thin] (2.710,0) rectangle (2.830,0.200) ;
\draw [fill=\typea,thin] (2.940,0) rectangle (3.060,1.000) ;
\draw [fill=\typeb,thin] (3.070,0) node[below]{\rotatebox[origin=t]{90}{\scriptsize web-BerkStan}} rectangle (3.190,1.000) ;
\draw [fill=\typec,thin] (3.200,0) rectangle (3.320,0.200) ;
\draw [fill=\typea,thin] (3.430,0) rectangle (3.550,0.020) ;
\draw [fill=\typeb,thin] (3.560,0) node[below]{\rotatebox[origin=t]{90}{\scriptsize web-Google}} rectangle (3.680,1.000) ;
\draw [fill=\typec,thin] (3.690,0) rectangle (3.810,0.800) ;
\draw [fill=\typea,thin] (3.920,0) rectangle (4.040,3.400) ;
\draw [fill=\typeb,thin] (4.050,0) node[below]{\rotatebox[origin=t]{90}{\scriptsize web-Stanford}} rectangle (4.170,1.200) ;
\draw [fill=\typec,thin] (4.180,0) rectangle (4.300,0.800) ;
\draw [fill=\typea,thin] (4.410,0) rectangle (4.530,0.400) ;
\draw [fill=\typeb,thin] (4.540,0) node[below]{\rotatebox[origin=t]{90}{\scriptsize wiki-Talk}} rectangle (4.660,1.400) ;
\draw [fill=\typec,thin] (4.670,0) rectangle (4.790,0.400) ;
\draw [fill=\typea,thin] (4.900,0) rectangle (5.020,0.400) ;
\draw [fill=\typeb,thin] (5.030,0) node[below]{\rotatebox[origin=t]{90}{\scriptsize youtube}} rectangle (5.150,0.800) ;
\draw [fill=\typec,thin] (5.160,0) rectangle (5.280,0.600) ;
\draw [fill=\typea,thin] (5.390,0) rectangle (5.510,1.200) ;
\draw [fill=\typeb,thin] (5.520,0) node[below]{\rotatebox[origin=t]{90}{\scriptsize flickr}} rectangle (5.640,0.800) ;
\draw [fill=\typec,thin] (5.650,0) rectangle (5.770,0.400) ;
\draw [fill=\typea,thin] (5.880,0) rectangle (6.000,1.400) ;
\draw [fill=\typeb,thin] (6.010,0) node[below]{\rotatebox[origin=t]{90}{\scriptsize livejournal}} rectangle (6.130,0.600) ;
\draw [fill=\typec,thin] (6.140,0) rectangle (6.260,0.020) ;
\draw [fill=\typea,thin] (6.570,0) rectangle (6.690,1.231) ;
\draw [fill=\typeb,thin] (6.700,0) node[below]{\rotatebox[origin=t]{90}{\scriptsize Average}} rectangle (6.820,1.000) ;
\draw [fill=\typec,thin] (6.830,0) rectangle (6.950,0.323) ;
\draw [<->, thick] (7.16,0) -- (0,0)-- (0,0.60) node[left] at (-0.8, 2.00) {\rotatebox{90}{\small Absolute error}} -- (0, 4.00) ;
\draw [dashed] (0, 1.000) node [left]{\scriptsize $0.005$} -- (6.980,1.000);
\draw [dashed] (0, 2.000) node [left]{\scriptsize$0.01$} -- (6.980,2.000);
\draw [dashed] (0, 3.000) node [left]{\scriptsize$0.015$} -- (6.980,3.000);
\draw [fill=\typea,thin] (1.05,3.60) node[right] at(1.41,3.72) {\small 2K} rectangle (1.29,3.84)  ;
\draw [fill=\typeb,thin] (2.46,3.60) node[right]at(2.81,3.72) {\small 8K} rectangle (2.70,3.84)  ;
\draw [fill=\typec,thin] (3.86,3.60) node[right]at(4.22,3.72) {\small 32K} rectangle (4.10,3.84)  ;
\end{tikzpicture}
}
}
 \caption{Local clustering coefficient results for increasing numbers of wedge samples. }
\label{fig:lcc}
\end{figure}
	
\section{Computing degree-wise clustering coefficients} 
\label{sec:deg-cc}

\begin{figure}[tb]
\centering
 
  \subfloat{\label{fig:ccd:amazon0505}
    \includegraphics[width=1.6in,trim=0 0 0 0]{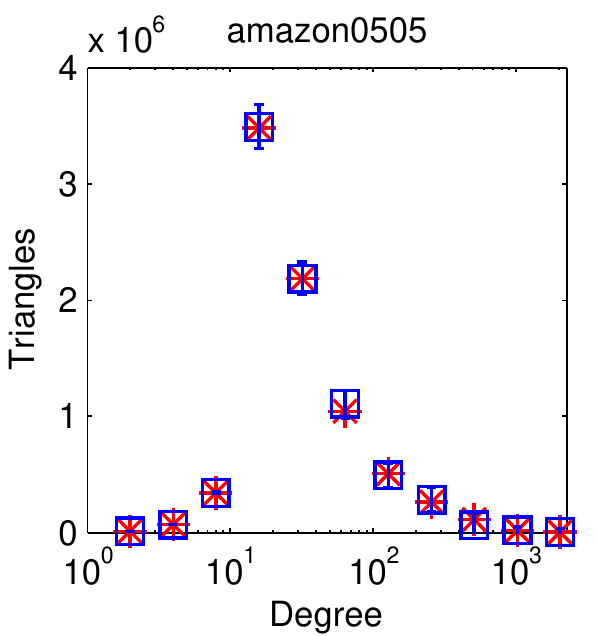}
  }
    \subfloat{\label{fig:ccd:amazon0601}
    \includegraphics[width=1.6in,trim=0 0 0 0]{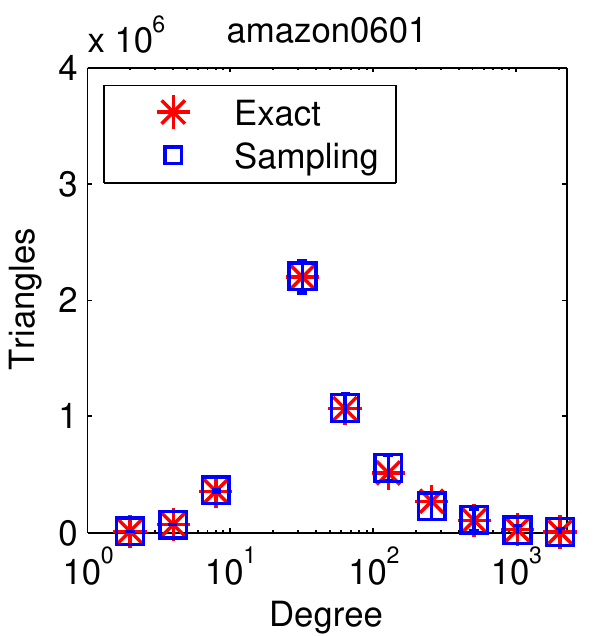}
  }
  \subfloat{\label{fig:ccd:as-skitter}
    \includegraphics[width=1.6in,trim=0 0 0 0]{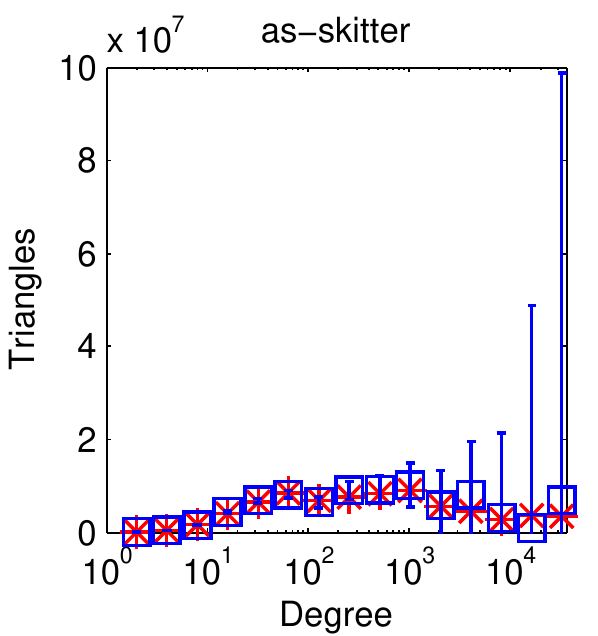}
  }
  \subfloat{\label{fig:ccd:cit-Patents}
    \includegraphics[width=1.6in,trim=0 0 0 0]{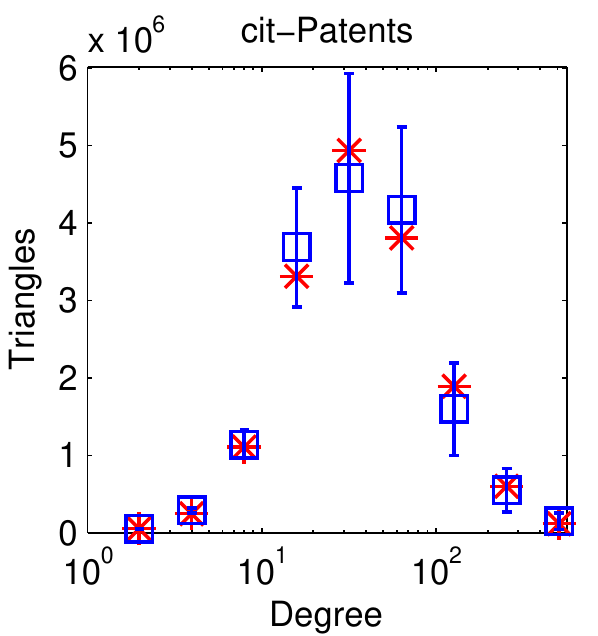}
  }  \\
  \subfloat{\label{fig:ccd:web-BerkStan}
    \includegraphics[width=1.6in,trim=0 0 0 0]{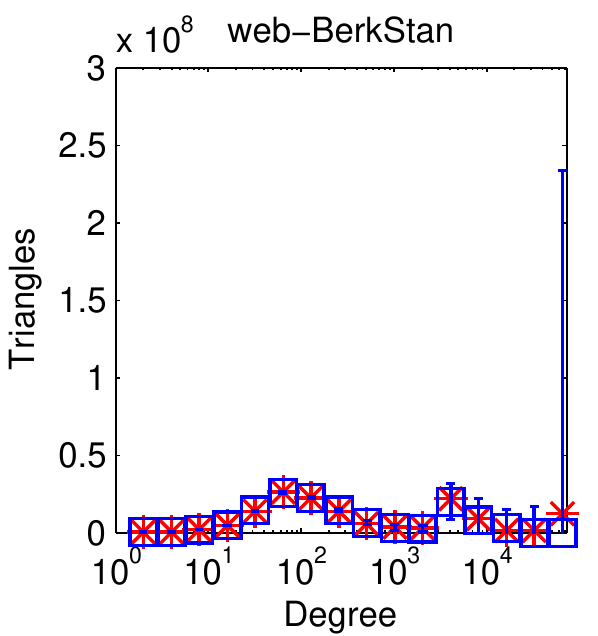}
  }
    \subfloat{\label{fig:ccd:web-Google}
    \includegraphics[width=1.6in,trim=0 0 0 0]{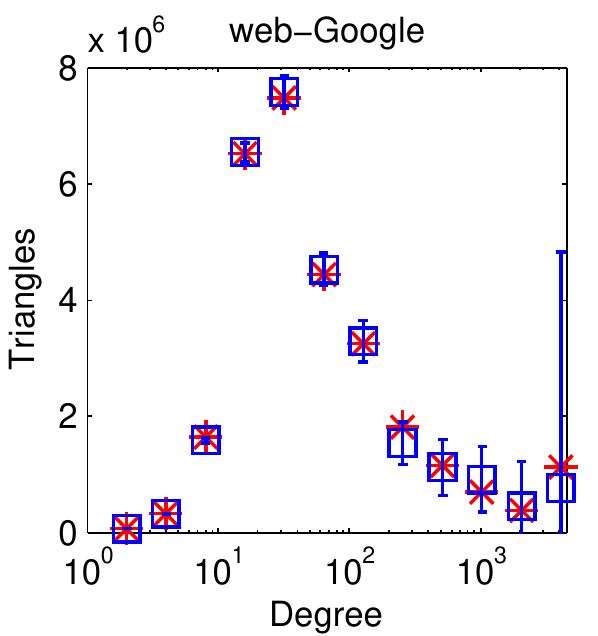}
  }
  \subfloat{\label{fig:ccd:web-Stanford}
    \includegraphics[width=1.6in,trim=0 0 0 0]{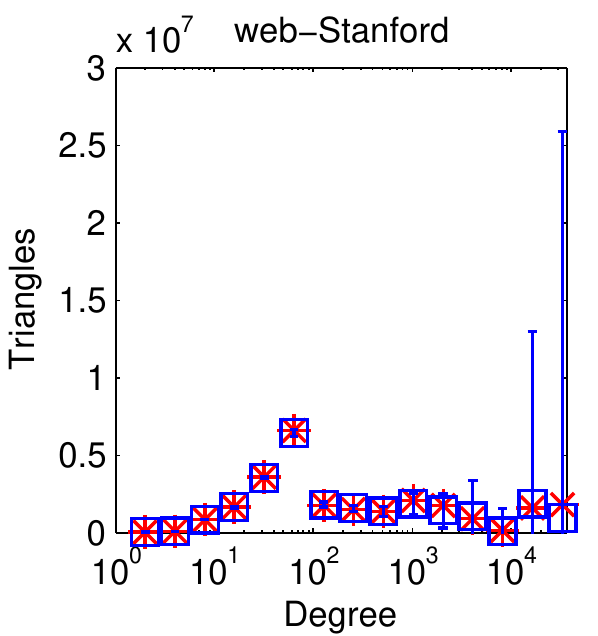}
  }
  \subfloat{\label{fig:ccd:wiki-Talk}
    \includegraphics[width=1.6in,trim=0 0 0 0]{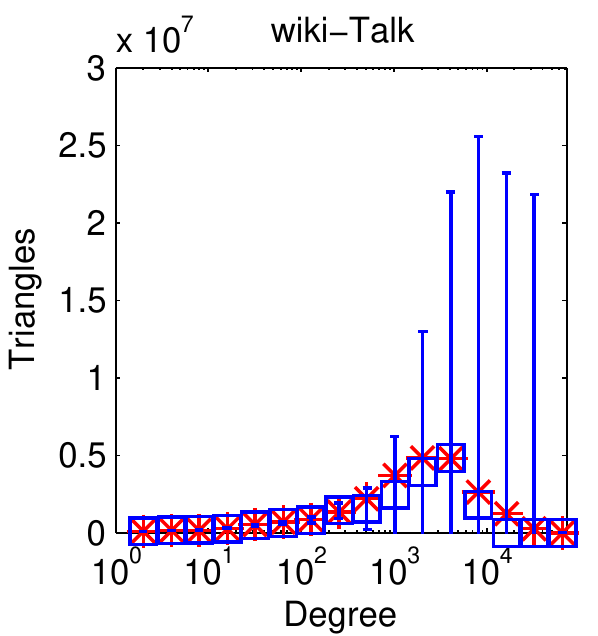}
  }     \\
    \subfloat{
    \includegraphics[width=1.6in,trim=0 0 0 0]{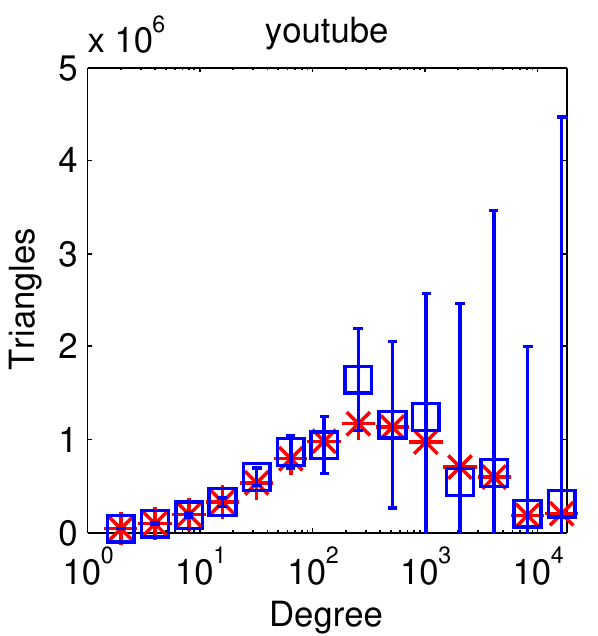}
  }
    \subfloat{
    \includegraphics[width=1.6in,trim=0 0 0 0]{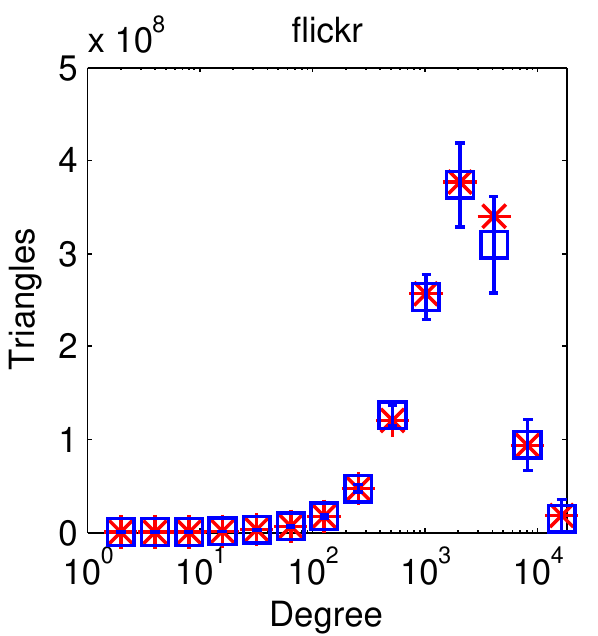}
  }
   \subfloat{
    \includegraphics[width=1.6in,trim=0 0 0 0]{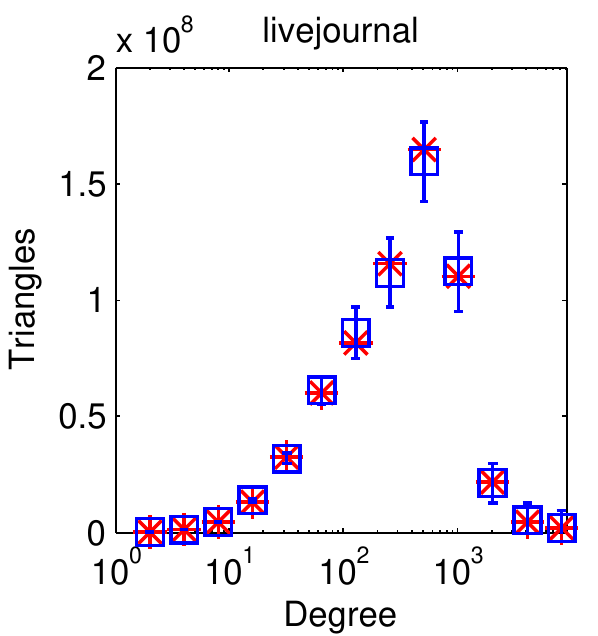}
  }
  
\caption{Computing degree-wise clustering coefficients using wedge
  sampling}
\label{fig:ccd}
\end{figure}

The real power of wedge sampling is demonstrated by estimating the
degree-wise clustering coefficients
$\{\dcc_d\}$. \Alg{dcc} shows procedure \algdcc.

\begin{algorithm}[htpb]
\caption{ Computing clustering coefficients by degree (\algdcc \label{alg:dcc})}
\begin{algorithmic}[1]
\STATE	 Pick $k$ uniform random vertices of degree $d$ (with replacement).
\STATE	 For each selected vertex $v$, choose a uniform random pair of neighbors of $v$ to generate a wedge.choose two
neighbors of $v$ (uniformly at random without replacement) to generate a random wedge. The set of all wedges comprises the sample set.
\STATE	 Output the fraction of closed wedges in the sample set as an estimate for $\dcc_d$.
\end{algorithmic}
\end{algorithm}

\begin{theorem} \label{thm:dcc} 
Set $k = \cei{0.5 \eps^{-2}\ln(2/\delta)}$. The algorithm \algdcc{} outputs an estimate
$\bar X$ for the clustering coefficient $\dcc_d$ such that
$|\bar{X} - \dcc_d| < \eps$
with probability greater than $(1-\delta)$.
\end{theorem}

\begin{proof} The proof is analogous to that
of \Thm{lcc}. Since $\dcc_d$ is the average clustering coefficient of a 
degree-$d$ vertex, we can apply the same arguments as in \Thm{lcc}.
\qed
\end{proof}

Algorithms in the previous section present how to compute the
clustering coefficient of vertices of a given degree.  In practice, it
may be sufficient to compute clustering coefficients over bins of
degrees.  Wedge sampling algorithms can be performed for bins of
degrees by a small adjustment of the sampling.
Within each bin, we weight each vertex according to the number of
wedges it produces. This 
guarantees that each wedge in the bin is equally likely to be selected. For instance, if we bin degree-3 and degree-4 vertices together, we will weight a degree-4 vertex twice as much as a degree 3-vertex, since a degree vertex generates ${3 \choose 2} =3$ wedges whereas a degree-4 vertex has  ${4 \choose 2} =6$. 
See \cite{KoPiPlSe13} for details of binned computations.

\Fig{ccd} shows results of \Alg{dcc} on three graphs for clustering coefficients.
In these experiments, we chose to group vertices in logarithmic
bins of degrees, i.e., $\set{2}, \set{3,4}, \set{5,6,7,8}, \dots$.
In other words, $2^{i-1} <d_v \leq 2^i$ form
the $i$-th bin.  The same algorithm  can be used for an arbitrary binning of the vertices.  
We show the estimates with increasing number of samples. At 8K
samples, the error is expected to be less than 0.02, which is
apparent in the plots.
Observe that even 500 samples yields a reasonable estimate in terms of
the differences by degree.

\Fig{ccdtimes} shows the time to calculate the binned $C_d$ values compared to enumeration; there is a
tremendous savings in runtime as a result of using wedge sampling. In
this figure, runtimes are normalized with respect to the runtimes of
full enumeration.  As the figure shows, wedge sampling takes only a
tiny fraction of the time of full enumeration, especially for large
graphs.

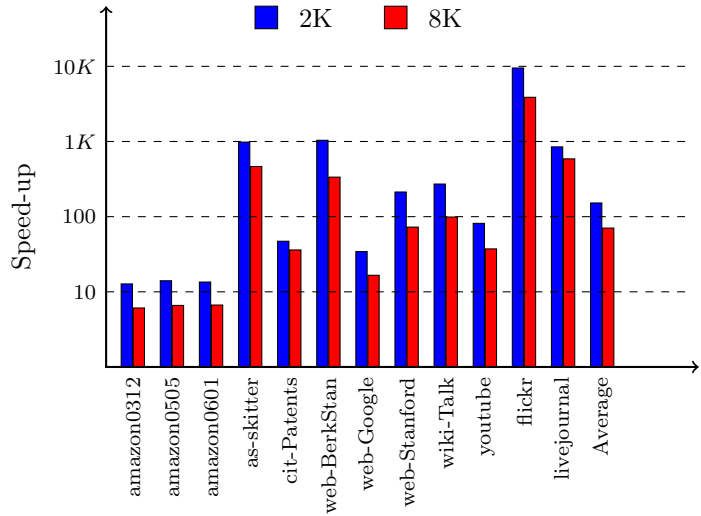
\begin{figure}[htpb]
  \centering
\begin{tikzpicture}\draw [fill=\typea,thin] (0.200,0) rectangle (0.350,1.105) ;
\draw [fill=\typeb,thin] (0.360,0) node[below]{\rotatebox[origin=t]{90}{\scriptsize amazon0312}} rectangle (0.510,0.785) ;
\draw [fill=\typea,thin] (0.720,0) rectangle (0.870,1.147) ;
\draw [fill=\typeb,thin] (0.880,0) node[below]{\rotatebox[origin=t]{90}{\scriptsize amazon0505}} rectangle (1.030,0.819) ;
\draw [fill=\typea,thin] (1.240,0) rectangle (1.390,1.130) ;
\draw [fill=\typeb,thin] (1.400,0) node[below]{\rotatebox[origin=t]{90}{\scriptsize amazon0601}} rectangle (1.550,0.824) ;
\draw [fill=\typea,thin] (1.760,0) rectangle (1.910,2.992) ;
\draw [fill=\typeb,thin] (1.920,0) node[below]{\rotatebox[origin=t]{90}{\scriptsize as-skitter}} rectangle (2.070,2.667) ;
\draw [fill=\typea,thin] (2.280,0) rectangle (2.430,1.673) ;
\draw [fill=\typeb,thin] (2.440,0) node[below]{\rotatebox[origin=t]{90}{\scriptsize cit-Patents}} rectangle (2.590,1.557) ;
\draw [fill=\typea,thin] (2.800,0) rectangle (2.950,3.016) ;
\draw [fill=\typeb,thin] (2.960,0) node[below]{\rotatebox[origin=t]{90}{\scriptsize web-BerkStan}} rectangle (3.110,2.526) ;
\draw [fill=\typea,thin] (3.320,0) rectangle (3.470,1.535) ;
\draw [fill=\typeb,thin] (3.480,0) node[below]{\rotatebox[origin=t]{90}{\scriptsize web-Google}} rectangle (3.630,1.220) ;
\draw [fill=\typea,thin] (3.840,0) rectangle (3.990,2.328) ;
\draw [fill=\typeb,thin] (4.000,0) node[below]{\rotatebox[origin=t]{90}{\scriptsize web-Stanford}} rectangle (4.150,1.861) ;
\draw [fill=\typea,thin] (4.360,0) rectangle (4.510,2.434) ;
\draw [fill=\typeb,thin] (4.520,0) node[below]{\rotatebox[origin=t]{90}{\scriptsize wiki-Talk}} rectangle (4.670,1.995) ;
\draw [fill=\typea,thin] (4.880,0) rectangle (5.030,1.911) ;
\draw [fill=\typeb,thin] (5.040,0) node[below]{\rotatebox[origin=t]{90}{\scriptsize youtube}} rectangle (5.190,1.571) ;
\draw [fill=\typea,thin] (5.400,0) rectangle (5.550,3.978) ;
\draw [fill=\typeb,thin] (5.560,0) node[below]{\rotatebox[origin=t]{90}{\scriptsize flickr}} rectangle (5.710,3.588) ;
\draw [fill=\typea,thin] (5.920,0) rectangle (6.070,2.928) ;
\draw [fill=\typeb,thin] (6.080,0) node[below]{\rotatebox[origin=t]{90}{\scriptsize livejournal}} rectangle (6.230,2.769) ;
\draw [fill=\typea,thin] (6.440,0) rectangle (6.590,2.182) ;
\draw [fill=\typeb,thin] (6.600,0) node[below]{\rotatebox[origin=t]{90}{\scriptsize Average}} rectangle (6.750,1.848) ;
\draw [fill=\typea,thin] (7.360,0) rectangle (7.510,0.000) ;
\draw [fill=\typeb,thin] (7.520,0) node[below]{\rotatebox[origin=t]{90}{\scriptsize }} rectangle (7.670,0.000) ;
\draw [<->, thick] (7.88,0) -- (0,0)-- (0,0.00) node[left] at (-0.8, 2.00) {\rotatebox{90}{\small Speed-up}} -- (0, 4.80) ;
\draw [dashed] (0, 1.000) node [left]{\scriptsize $10$} -- (7.700,1.000);
\draw [dashed] (0, 2.000) node [left]{\scriptsize$100$} -- (7.700,2.000);
\draw [dashed] (0, 3.000) node [left]{\scriptsize$1K$} -- (7.700,3.000);
\draw [dashed] (0, 4.000) node [left]{\scriptsize$10K$} -- (7.700,4.000);
\draw [fill=\typea,thin] (1.98,4.50) node[right] at(2.43,4.65) {\small 2K} rectangle (2.28,4.80)  ;
\draw [fill=\typeb,thin] (3.70,4.50) node[right]at(4.15,4.65) {\small 8K} rectangle (4.00,4.80)  ;
\end{tikzpicture}
  \caption{Speed-up in binned degree-wise clustering coefficient computation time for increasing numbers of wedge samples}
  \label{fig:ccdtimes}
\end{figure}

\section{Counting undirected triangles per degree}
\label{sec:trid}

By modifying the template given in \Sec{wedge}, we can also get estimates
for $T_d$ (the number of triangles
incident to degree-$d$ vertices). 
Instead of counting the fraction of closed wedges,
we take a weighted sum. \Alg{trid} describes the procedure \algtrid.
We let
$W_d = n_d \cdot {d \choose 2}$ denote the total number of wedges centered at degree-$d$ vertices.

\begin{algorithm}[htbp]
  \caption{ Computing triangles by degree (\algtrid{}) \label{alg:trid}}
  \begin{algorithmic}[1]
    \STATE Pick $k$ uniform random vertices of degree $d$ (with replacement).
    \STATE For each selected vertex $v$, choose two
neighbors of $v$ (uniformly at random without replacement) to generate a random wedge. The set of all wedges comprises the sample set.
    \STATE For each wedge $w_i$ in the sample set, let $Y_i$ be the associated random
    variable such that
    \begin{displaymath}
    Y_i =
    \begin{cases}
    0 & \text{if $w$ is open},\\
    \frac{1}{3} & \text{if $w$ is closed and has 3 vertices in $V_d$},\\
    \frac{1}{2} & \text{if $w$ is closed and has 2 vertices in $V_d$},\\
    1 & \text{if $w$ is closed and has 1 vertex in $V_d$}.
    \end{cases}
    \end{displaymath}
    \STATE Set $\bar Y = \frac{1}{k} \sum_i Y_i$.
    \STATE Output $W_d \cdot \bar{Y}$ as the estimate for $T_d$.
  \end{algorithmic}
\end{algorithm}
\medskip

\begin{theorem} \label{thm:dcc2} 
  Set $k = \cei{0.5 \eps^{-2}\ln(2/\delta)}$. 
  The algorithm \algtrid{} outputs an estimate
  $W_d \cdot \bar Y$ for the $T_d$ with the following guarantee:
  $|W_d \cdot \bar{Y} - T_d| < \eps W_d$
  with probability greater than $1-\delta$.
\end{theorem}

\begin{proof} For a single sampled wedge $w_i$, we define $Y_i$. We will
  argue below that the expected value of $\EX[Y]$ is exactly $T_d/W_d$. Once we have that, an application
  of the Hoeffding bound of \Thm{Hoeffding} shows that $|\bar Y - T_d/W_d| < \eps$ with
  probability greater than $1-\delta$. Multiplying this inequality by
  $W_d$, we get $|W_d \cdot \bar{Y} - T_d| < \eps W_d$,
  completing the proof.
  
  To show $\EX[Y] = T_d/W_d$, partition the set of all wedges centered on degree $d$ vertices into four
  sets $S_0, S_1, S_2$, and $S_3$. The set $S_i$ ($i \neq 0$) contains all
  closed wedges containing
  exactly $i$ degree-$d$ vertices. The remaining open wedges
  go into $S_0$. For a sampled wedge $w$, if $w \in S_i, i \neq 0$, then
  $Y_i = 1/i$. If $w \in S_0$, then $Y_i = 0$. The wedge $w$ is 
  a uniform random wedge from those centered on degree-$d$ vertices. Hence,
  $\EX[Y] = W^{-1}_d(|S_1| + |S_2|/2 + |S_3|/3)$. 
  
  Now partition the set of triangles involving degree $d$ vertices into
  three sets $S'_1, S'_2, S'_3$, where $S'_i$ is the set of triangles
  with $i$ degree $d$ vertices. Observe that $|S_i| = i|S'_i|$. If a triangle
  has $i$ vertices of degree $d$, then there are exactly $i$ wedges
  centered in degree $d$ vertices (in that triangle). 
  So, $|S_1| + |S_2|/2 + |S_3|/3 = $ $|S'_1| + |S'_2| + |S'_3|$ $= T_d$. Therefore,
  $\EX[Y] = T_d/W_d$.
  \qed
\end{proof}

The results of \Alg{trid} are shown in \Fig{td}. Once again, 
the data is grouped in logarithmic
bins of degrees, i.e., $\set{2}, \set{3,4}, \set{5,6,7,8}, \dots$.
In other words, $2^{i-1} <d_v \leq 2^i$ form
the $i$-th bin.  
The number of samples is $k=2000$. Here we also show the error bars corresponding to 99\% confidence ($\delta=0.01$).

\begin{figure}[hbtp]
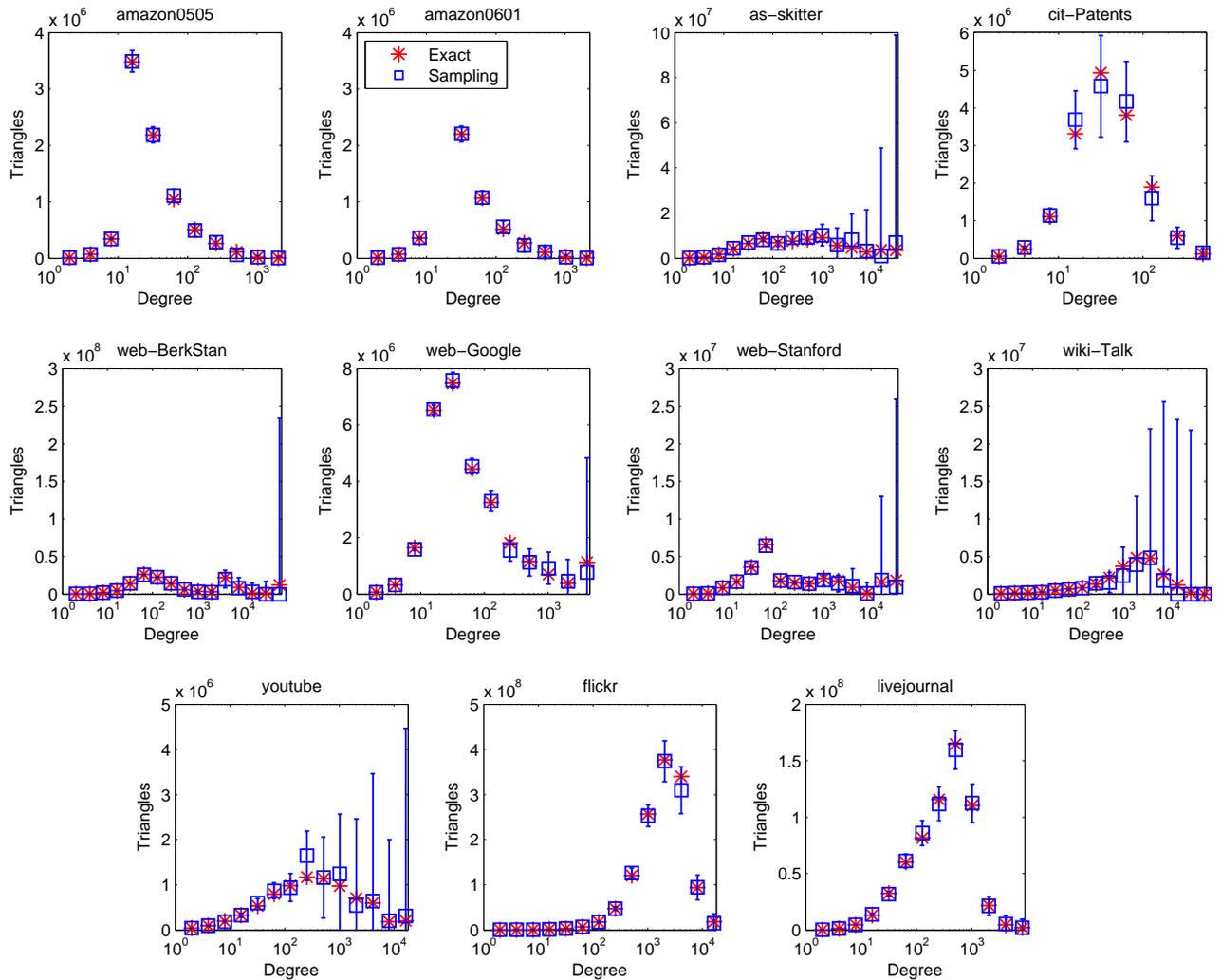

  \centering
    \subfloat{\label{fig:td:amazon0505}
    \includegraphics[width=1.6in,trim=0 0 0 0]{amazon0505-eps-converted-to.pdf}
  }
    \subfloat{\label{fig:td:amazon0601}
    \includegraphics[width=1.6in,trim=0 0 0 0]{amazon0601-eps-converted-to.pdf}
  }
  \subfloat{\label{fig:td:as-skitter}
    \includegraphics[width=1.6in,trim=0 0 0 0]{as-skitter-eps-converted-to.pdf}
  }
  \subfloat{\label{fig:td:cit-Patents}
    \includegraphics[width=1.6in,trim=0 0 0 0]{cit-Patents-eps-converted-to.pdf}
  } \\  
  \subfloat{\label{fig:td:web-BerkStan}
    \includegraphics[width=1.6in,trim=0 0 0 0]{web-BerkStan-eps-converted-to.pdf}
  }
    \subfloat{\label{fig:td:web-Google}
    \includegraphics[width=1.6in,trim=0 0 0 0]{web-Google-eps-converted-to.pdf}
  }
  \subfloat{\label{fig:td:web-Stanford}
    \includegraphics[width=1.6in,trim=0 0 0 0]{web-Stanford-eps-converted-to.pdf}
  }
  \subfloat{\label{fig:td:wiki-Talk}
    \includegraphics[width=1.6in,trim=0 0 0 0]{wiki-Talk-eps-converted-to.pdf}
  }     \\
  \subfloat{
    \includegraphics[width=1.6in,trim=0 0 0 0]{youtube-eps-converted-to.pdf}
  }
  \subfloat{
    \includegraphics[width=1.6in,trim=0 0 0 0]{flickr-eps-converted-to.pdf}
  }
  \subfloat{
    \includegraphics[width=1.6in,trim=0 0 0 0]{livejournal-eps-converted-to.pdf}
  }
  
\caption{Number of triangles per degree bin. The red asterisk is the exact count, while the blue box is the estimate based on 2000 samples. The bars show the error at 99\% confidence.}
\label{fig:td}
\end{figure}

\section{Counting directed triangles}
\label{sec:directed}

Counting  triangles in directed graphs is  more difficult, because  there are six types of wedges  and seven types of directed triangles  (up  to isomorphism); see \Fig{directed}.  
As mentioned earlier, a reciprocal edge is a pair of one-way edges $\{(i,j),(j,i)\}$ that are treated as a \emph{single}
reciprocal edge. For each vertex $v$, we have three associated degrees: the indegree, outdegree, and reciprocal degree, denoted by $\din{v}$, $\dout{v}$, and $\drec{v}$, respectively. 

We can generalize wedge sampling to count the different triangles types in~\Fig{dtri}. 
This is done by randomly sampling the different wedges, and looking for various directed closures.
We need some notation to formalize various directed triangle concepts, which is adapted from~\cite{SePiDuKo13}.

\begin{figure}[tbhp]
\centering
\subfloat[Directed wedges]
{
  \begin{tikzpicture}[nd/.style={circle,draw,fill=teal!50,inner sep=2pt},framed]    
    \matrix[column sep=0.4cm, row sep=0.2cm,ampersand replacement=\&]
    {
      \directedwedge{\linkba}{\linkba}{(i)} \&
      \directedwedge{\linkab}{\linkba}{(ii)} \&
      \directedwedge{\linkab}{\linkab}{(iii)} \\
      \directedwedge{\linkaba}{\linkab}{(iv)} \&
      \directedwedge{\linkaba}{\linkba}{(v)} \&
      \directedwedge{\linkaba}{\linkaba}{(vi)} \\
    };
  \end{tikzpicture}  \label{fig:dwedge}}
  \quad 
  \subfloat[Directed triangles]
  {\label{fig:dtri}
  \begin{tikzpicture}[nd/.style={circle,draw,fill=teal!50,inner sep=2pt},framed]    
    \matrix[column sep=0.4cm, row sep=0.2cm,ampersand replacement=\&]
    {
      \directedtriangle{\linkab}{\linkba}{\linkba}{(a)} \&
       \directedtriangle{\linkab}{\linkba}{\linkab}{(b)}  \&
       \directedtriangle{\linkaba}{\linkba}{\linkba}{(c)} \&
       \directedtriangle{\linkaba}{\linkab}{\linkba}{(d)} \\ 
      \directedtriangle{\linkaba}{\linkab}{\linkab}{(e)} \&
      \directedtriangle{\linkaba}{\linkaba}{\linkba}{(f)} \&
      \directedtriangle{\linkaba}{\linkaba}{\linkaba}{(g)} \& \\
    };
  \end{tikzpicture}  }

\subfloat[Number of directed wedges for each wedge type]{
\label{fig:wedge_counts}
\begin{tabular}{c|cccccc} 
$\psi$ %
& i & ii & iii & iv & v & vi 
\\\hline
$W_{v,\psi}$  
& $\displaystyle {\dout{v} \choose 2}$
& $\displaystyle \din{v}\dout{v}$ 
& $\displaystyle {\din{v} \choose 2}$   
& $\displaystyle \din{v}\drec{v}$ 
& $\displaystyle \dout{v}\drec{v}$ 
& $\displaystyle {\drec{v} \choose 2}$   \\
\end{tabular}
}
\\
\subfloat[The values of $\chi(\psi,\ttype)$, indicating of many wedges of type $\psi$ participate in triangles of type $\ttype$.\label{fig:triangle_by_wedge} ]{
\setlength{\tabcolsep}{3pt}
\begin{tabular}{c|c|cccccc|} 
\multicolumn{2}{c}{}
 & \multicolumn{6}{c}{Wedge type ($\psi$)} \\ \cline{3-8}
 \multicolumn{1}{c}{} & \multicolumn{1}{c|}{} & i & ii & iii  & iv  & v  & v \\ \cline{2-8}
\multirow{6}{*}{\vspace{-1.3in}\begin{rotate}{90}Triangle type ($\ttype$)\end{rotate}}
 & a & 1 & \textbf{1} & 1 &   &   &   \\ %
 & b &   & \textbf{3} &   &   &   &   \\ %
 & c & 1 &   &   &   & \textbf{2} &   \\ %
 & d &   & 1 &   & \textbf{1} & 1 &   \\ %
 & e &   &   & 1 & \textbf{2} &   &   \\ %
 & f &   &   &   & 1 & \textbf{1} & 1 \\ %
 & g &   &   &   &   &   & \textbf{3} \\ 
 \cline{2-8} %
\end{tabular} 
}
\caption{Information for counting triangles in directed graphs.}
\label{fig:directed}
\end{figure}

Naturally, the first step is to construct methods for sampling the different wedge types (see \Fig{dwedge}).
For wedge type $\psi$, let
$W_\psi$ denote the number of wedges of this type. 
For vertex $v$, let $W_{v,\psi}$ be the number of $\psi$-wedges
centered at $v$. It is straightforward to compute $W_{v,\psi}$ given the degrees of $v$; 
see the formulas in \Fig{wedge_counts}.

Given the number of $\psi$-wedges per vertex, we can use our standard template to sample
a uniform random $\psi$-wedge. This is given in the first two steps of \Alg{dtri}.
The procedure to select a uniform random $\psi$-wedge centered at $v$ depends on 
the wedge type. For example, to sample a type-\wout{} wedge, we pick a uniform random
pair of out-neighbors of $v$ (without replacement). To sample a type-\wmid{} wedge, 
we pick a uniform
random out-neighbor, and a uniform random in-neighbor. All other wedge types
are sampled analogously.

To get the triangle counts, we need a directed transitivity that measures
the fraction of closed directed wedges. Of course, a directed wedge can close
into many different types of triangles. For example, an type-\wout wedge can close
into a type-\ttrans{} or type-\toutrecip triangle. To measure this meaningfully,
a set of directed closure measures are defined in~\cite{SePiDuKo13}. 
We restate these definitions and show how
to count triangles using these measures.

Let $\chi(\psi,\ttype)$ denote the number of type-$\psi$ wedges in a type-$\ttype$ triangle.
For example, a type-\ttrans triangle
contains exactly one type-\wout wedge, but a type-\tcycle triangle contains three type-\wmid  wedges.
The list of these values is given in \Fig{triangle_by_wedge}.

We say that a $\psi$-wedge
is \emph{$\ttype$-closed} if the wedge participates in a $\ttype$-triangle.
The number of such wedges is exactly $\chi(\psi,\tau) T_\ttype$
(where $T_\ttype$ is the number of $\ttype$-triangles).
\emph{The $(\psi,\ttype)$-closure, $\gcc_{\psi,\ttype}$,
is the fraction of $\psi$-wedges that are $\ttype$-closed.}
$$ \gcc_{\psi,\ttype} = \frac{\chi(\psi,\ttype) T_\ttype}{W_\psi} $$
There are 15 non-trivial $(\psi,\ttype)$-closures, corresponding to the
non-zero entries in \Fig{triangle_by_wedge}. By the wedge sampling framework, 
we can estimate any $\gcc_{\psi,\ttype}$ value.
This value can be used to estimate $T_{\ttype}$.

\begin{algorithm}
\caption{Computing directed closure ratios (\algdtri{})\label{alg:dtri}}
\begin{algorithmic}[1]
    \STATE Select $k$ random vertices (with replacement) according to the
probability distribution defined by $\{p_v\}$ where $p_v =  W_{v,\psi}/W_\psi$.
    \STATE For each selected vertex $v$, pick a uniform random $\psi$-wedge centered at $v$. Note that if both edges are of the same type (i.e., two out-edges), then sampling is done without replacement. Otherwise, since the two edge-types are distinct and the sampling is independent.
	\STATE Determine $k'$, the number of $\ttype$-closed wedges among the sampled wedges.
	\STATE Compute estimate $\widehat{\gcc}_{\psi,\ttype} = k'/k$ for $\gcc_{\psi,\ttype}$.
	\STATE Output estimate $\widehat{\gcc}_{\psi,\ttype}\cdot W_\psi/\chi(\psi,\ttype)$ for $T_\ttype$.
	\end{algorithmic}
\end{algorithm}
\medskip

\begin{theorem} \label{thm:main-dir}
Fix a triangle type $\ttype$ and wedge type $\psi$ such that $\chi(\psi,\ttype) \neq 0$. 
Set $k = \cei{0.5 \eps^{-2}\ln(2/\delta)}$. The algorithm \algdtri{}
outputs an estimate 
$\widehat{T}_{\ttype}$ for $T_{\ttype}$
such that $|\widehat{T}_{\ttype} - T_{\ttype}| < \eps W_\psi/\chi(\psi,\ttype)$
with probability greater than $(1-\delta)$.
\end{theorem}

\begin{proof} Using the Hoeffding bound, we can argue (as before)
that $|\widehat{\gcc}_{\psi,\ttype} - \gcc_{\psi,\ttype}| < \eps$ with probability at least $(1-\delta)$.
We multiply this inequality by $W_\psi/\chi(\psi,\ttype)$ and observe that $T_\ttype = \gcc_{\psi,\ttype}\cdot W_\psi/\chi(\psi,\ttype)$.
\end{proof}

\begin{table}[thbp]
\caption{Properties of the directed graphs}
\label{tab:properties}
\centering
{\setlength{\tabcolsep}{1em}
\begin{tabular}{|r|r|r|r|r|r|r|}

\hline
\multicolumn{1}{|c|}{
Graph Name} &\multicolumn{1}{c|}{Nodes}& 	\multicolumn{1}{c|}{Dir.~Edges} & 	\multicolumn{1}{c|}{Wedges}	
&\multicolumn{1}{c|}{Triangles} & \multicolumn{1}{c|}{r} & \multicolumn{1}{c|}{$\gcc$}   \\\hline  	
amazon0505	& 410K&	3357K& 73M&	3951K & 0.55	&	0.162	\\
web-Stanford	&282K&	2312K &3944M & 11330K&	0.28 & 0.009\\
web-BerkStan&	685K	&7601K&	27983M &	64691K &	0.25 &	0.007\\
wiki-Talk	&2394K&	5021K&	12594M&	9204K	& 0.14 & 0.002 \\
web-Google&	876K&	5105K&	727M& 13392K&	0.31&	0.055\\
youtube&	1158K &	4945K&	1474M&	3057K&	0.79&	0.006 \\
flickr	&1861K &	22614K&\ 	14675M&	548659K &	0.62&	0.112\\
livejournal& 5284K & 76938K & 	7519M & 	310877K &	0.73 &  0.124\\ \hline 
\end{tabular}
}
\end{table} 

The experimental setup is the same as in the previous sections.  
We performed our experiments on 8 directed graphs, described in \Tab{properties}.
The results of applying \Alg{dtri} are shown in 
\Fig{dtd}. We use 32K wedges samples for each
of four specific wedge types (detailed below). As can be seen, we get accurate results for 
all triangle types, except possibly type-(b) triangles. These are so few that they do not even appear in the plot.

\begin{figure}[tbhp]
\centering
  \subfloat[amazon0505]{\label{fig:d-amazon0505}
    %
\begin{tikzpicture}[scale=1.000000]\draw [fill=\typea,thin] (0.100,0) node[below] at (0.250,0){{\tiny (a)}} rectangle (0.250,1.966) ;
\draw [fill=\typeb,thin] (0.260,0) rectangle (0.410,1.893) ;
\draw [fill=\typea,thin] (0.520,0) node[below] at (0.670,0){{\tiny (b)}} rectangle (0.670,0.002) ;
\draw [fill=\typeb,thin] (0.680,0) rectangle (0.830,0.002) ;
\draw [fill=\typea,thin] (0.940,0) node[below] at (1.090,0){{\tiny (c)}} rectangle (1.090,2.876) ;
\draw [fill=\typeb,thin] (1.100,0) rectangle (1.250,2.877) ;
\draw [fill=\typea,thin] (1.360,0) node[below] at (1.510,0){{\tiny (d)}} rectangle (1.510,0.201) ;
\draw [fill=\typeb,thin] (1.520,0) rectangle (1.670,0.198) ;
\draw [fill=\typea,thin] (1.780,0) node[below] at (1.930,0){{\tiny (e)}} rectangle (1.930,1.871) ;
\draw [fill=\typeb,thin] (1.940,0) rectangle (2.090,1.909) ;
\draw [fill=\typea,thin] (2.200,0) node[below] at (2.350,0){{\tiny (f)}} rectangle (2.350,2.427) ;
\draw [fill=\typeb,thin] (2.360,0) rectangle (2.510,2.316) ;
\draw [fill=\typea,thin] (2.620,0) node[below] at (2.770,0){{\tiny (g)}} rectangle (2.770,2.511) ;
\draw [fill=\typeb,thin] (2.780,0) rectangle (2.930,2.526) ;
\draw [<->, thick] (3.24,0) -- (0,0)-- (0,0.00) node[left] at (-0.8, 1.50) {\rotatebox{90}{\tiny \# Triangles}} -- (0, 3.00) ;
\draw [dashed] (0, 0.750) node [left]{\tiny $250K$} -- (3.060,0.750);
\draw [dashed] (0, 1.500) node [left]{\tiny $500K$} -- (3.060,1.500);
\draw [dashed] (0, 2.250) node [left]{\tiny $750K$} -- (3.060,2.250);
\draw [fill=\typea,thin] (0.50,3.00) node[right] at(0.65,3.15) {\tiny Exact} rectangle (0.72,3.23)  ;
\draw [fill=\typeb,thin] (2.00,3.00) node[right]at(2.15,3.15) {\tiny 32K} rectangle (2.23,3.23)  ;
\end{tikzpicture}
  }
\subfloat[web-Stanford]{\label{fig:d-Stanford}
    %
\begin{tikzpicture}[scale=1.000000]\draw [fill=\typea,thin] (0.100,0) node[below] at (0.250,0){{\tiny (a)}} rectangle (0.250,2.800) ;
\draw [fill=\typeb,thin] (0.260,0) rectangle (0.410,2.810) ;
\draw [fill=\typea,thin] (0.520,0) node[below] at (0.670,0){{\tiny (b)}} rectangle (0.670,0.005) ;
\draw [fill=\typeb,thin] (0.680,0) rectangle (0.830,0.005) ;
\draw [fill=\typea,thin] (0.940,0) node[below] at (1.090,0){{\tiny (c)}} rectangle (1.090,1.895) ;
\draw [fill=\typeb,thin] (1.100,0) rectangle (1.250,1.915) ;
\draw [fill=\typea,thin] (1.360,0) node[below] at (1.510,0){{\tiny (d)}} rectangle (1.510,0.027) ;
\draw [fill=\typeb,thin] (1.520,0) rectangle (1.670,0.026) ;
\draw [fill=\typea,thin] (1.780,0) node[below] at (1.930,0){{\tiny (e)}} rectangle (1.930,0.620) ;
\draw [fill=\typeb,thin] (1.940,0) rectangle (2.090,0.620) ;
\draw [fill=\typea,thin] (2.200,0) node[below] at (2.350,0){{\tiny (f)}} rectangle (2.350,0.130) ;
\draw [fill=\typeb,thin] (2.360,0) rectangle (2.510,0.141) ;
\draw [fill=\typea,thin] (2.620,0) node[below] at (2.770,0){{\tiny (g)}} rectangle (2.770,0.183) ;
\draw [fill=\typeb,thin] (2.780,0) rectangle (2.930,0.184) ;
\draw [<->, thick] (3.24,0) -- (0,0)-- (0,0.00) node[left] at (-0.8, 1.50) {\rotatebox{90}{\tiny \# Triangles}} -- (0, 3.00) ;
\draw [dashed] (0, 0.750) node [left]{\tiny $1.5M$} -- (3.060,0.750);
\draw [dashed] (0, 1.500) node [left]{\tiny $3M$} -- (3.060,1.500);
\draw [dashed] (0, 2.250) node [left]{\tiny $4.5M$} -- (3.060,2.250);
\draw [fill=\typea,thin] (0.50,3.00) node[right] at(0.65,3.15) {\tiny Exact} rectangle (0.72,3.23)  ;
\draw [fill=\typeb,thin] (2.00,3.00) node[right]at(2.15,3.15) {\tiny 32K} rectangle (2.23,3.23)  ;
\end{tikzpicture}
  } 
\subfloat[web-BerkStan]{\label{fig:d-BerkStan}
    %
\begin{tikzpicture}[scale=1.000000]\draw [fill=\typea,thin] (0.100,0) node[below] at (0.250,0){{\tiny (a)}} rectangle (0.250,3.000) ;
\draw [fill=\typeb,thin] (0.260,0) rectangle (0.410,3.000) ;
\draw [fill=\typea,thin] (0.520,0) node[below] at (0.670,0){{\tiny (b)}} rectangle (0.670,0.003) ;
\draw [fill=\typeb,thin] (0.680,0) rectangle (0.830,0.003) ;
\draw [fill=\typea,thin] (0.940,0) node[below] at (1.090,0){{\tiny (c)}} rectangle (1.090,2.875) ;
\draw [fill=\typeb,thin] (1.100,0) rectangle (1.250,2.875) ;
\draw [fill=\typea,thin] (1.360,0) node[below] at (1.510,0){{\tiny (d)}} rectangle (1.510,0.032) ;
\draw [fill=\typeb,thin] (1.520,0) rectangle (1.670,0.031) ;
\draw [fill=\typea,thin] (1.780,0) node[below] at (1.930,0){{\tiny (e)}} rectangle (1.930,1.288) ;
\draw [fill=\typeb,thin] (1.940,0) rectangle (2.090,1.288) ;
\draw [fill=\typea,thin] (2.200,0) node[below] at (2.350,0){{\tiny (f)}} rectangle (2.350,0.096) ;
\draw [fill=\typeb,thin] (2.360,0) rectangle (2.510,0.096) ;
\draw [fill=\typea,thin] (2.620,0) node[below] at (2.770,0){{\tiny (g)}} rectangle (2.770,0.798) ;
\draw [fill=\typeb,thin] (2.780,0) rectangle (2.930,0.799) ;
\draw [<->, thick] (3.24,0) -- (0,0)-- (0,0.00) node[left] at (-0.8, 1.50) {\rotatebox{90}{\tiny \# Triangles}} -- (0, 3.00) ;
\draw [dashed] (0, 0.750) node [left]{\tiny $6M$} -- (3.060,0.750);
\draw [dashed] (0, 1.500) node [left]{\tiny $12M$} -- (3.060,1.500);
\draw [dashed] (0, 2.250) node [left]{\tiny $18M$} -- (3.060,2.250);
\draw [fill=\typea,thin] (0.50,3.00) node[right] at(0.65,3.15) {\tiny Exact} rectangle (0.72,3.23)  ;
\draw [fill=\typeb,thin] (2.00,3.00) node[right]at(2.15,3.15) {\tiny 32K} rectangle (2.23,3.23)  ;
\end{tikzpicture}
  }
  
 \subfloat[wiki-Talk]{\label{fig:d-wiki-Talk}
    %
\begin{tikzpicture}[scale=1.000000]\draw [fill=\typea,thin] (0.100,0) node[below] at (0.250,0){{\tiny (a)}} rectangle (0.250,2.280) ;
\draw [fill=\typeb,thin] (0.260,0) rectangle (0.410,2.260) ;
\draw [fill=\typea,thin] (0.520,0) node[below] at (0.670,0){{\tiny (b)}} rectangle (0.670,0.172) ;
\draw [fill=\typeb,thin] (0.680,0) rectangle (0.830,0.144) ;
\draw [fill=\typea,thin] (0.940,0) node[below] at (1.090,0){{\tiny (c)}} rectangle (1.090,1.010) ;
\draw [fill=\typeb,thin] (1.100,0) rectangle (1.250,1.020) ;
\draw [fill=\typea,thin] (1.360,0) node[below] at (1.510,0){{\tiny (d)}} rectangle (1.510,1.060) ;
\draw [fill=\typeb,thin] (1.520,0) rectangle (1.670,0.883) ;
\draw [fill=\typea,thin] (1.780,0) node[below] at (1.930,0){{\tiny (e)}} rectangle (1.930,1.610) ;
\draw [fill=\typeb,thin] (1.940,0) rectangle (2.090,1.620) ;
\draw [fill=\typea,thin] (2.200,0) node[below] at (2.350,0){{\tiny (f)}} rectangle (2.350,2.230) ;
\draw [fill=\typeb,thin] (2.360,0) rectangle (2.510,2.160) ;
\draw [fill=\typea,thin] (2.620,0) node[below] at (2.770,0){{\tiny (g)}} rectangle (2.770,0.836) ;
\draw [fill=\typeb,thin] (2.780,0) rectangle (2.930,0.838) ;
\draw [<->, thick] (3.24,0) -- (0,0)-- (0,0.00) node[left] at (-0.8, 1.50) {\rotatebox{90}{\tiny \# Triangles}} -- (0, 3.00) ;
\draw [dashed] (0, 0.750) node [left]{\tiny $0.75M$} -- (3.060,0.750);
\draw [dashed] (0, 1.500) node [left]{\tiny $1.5M$} -- (3.060,1.500);
\draw [dashed] (0, 2.250) node [left]{\tiny $2.25M$} -- (3.060,2.250);
\draw [fill=\typea,thin] (0.50,3.00) node[right] at(0.65,3.15) {\tiny Exact} rectangle (0.72,3.23)  ;
\draw [fill=\typeb,thin] (2.00,3.00) node[right]at(2.15,3.15) {\tiny 32K} rectangle (2.23,3.23)  ;
\end{tikzpicture}
  }
      \subfloat[web-Google]{\label{fig:d-web-Google}
    %
\begin{tikzpicture}[scale=1.000000]\draw [fill=\typea,thin] (0.100,0) node[below] at (0.250,0){{\tiny (a)}} rectangle (0.250,2.868) ;
\draw [fill=\typeb,thin] (0.260,0) rectangle (0.410,2.880) ;
\draw [fill=\typea,thin] (0.520,0) node[below] at (0.670,0){{\tiny (b)}} rectangle (0.670,0.021) ;
\draw [fill=\typeb,thin] (0.680,0) rectangle (0.830,0.017) ;
\draw [fill=\typea,thin] (0.940,0) node[below] at (1.090,0){{\tiny (c)}} rectangle (1.090,2.514) ;
\draw [fill=\typeb,thin] (1.100,0) rectangle (1.250,2.490) ;
\draw [fill=\typea,thin] (1.360,0) node[below] at (1.510,0){{\tiny (d)}} rectangle (1.510,0.119) ;
\draw [fill=\typeb,thin] (1.520,0) rectangle (1.670,0.124) ;
\draw [fill=\typea,thin] (1.780,0) node[below] at (1.930,0){{\tiny (e)}} rectangle (1.930,1.164) ;
\draw [fill=\typeb,thin] (1.940,0) rectangle (2.090,1.158) ;
\draw [fill=\typea,thin] (2.200,0) node[below] at (2.350,0){{\tiny (f)}} rectangle (2.350,0.510) ;
\draw [fill=\typeb,thin] (2.360,0) rectangle (2.510,0.495) ;
\draw [fill=\typea,thin] (2.620,0) node[below] at (2.770,0){{\tiny (g)}} rectangle (2.770,0.840) ;
\draw [fill=\typeb,thin] (2.780,0) rectangle (2.930,0.846) ;
\draw [<->, thick] (3.24,0) -- (0,0)-- (0,0.00) node[left] at (-0.8, 1.50) {\rotatebox{90}{\tiny \# Triangles}} -- (0, 3.00) ;
\draw [dashed] (0, 0.750) node [left]{\tiny $1.25M$} -- (3.060,0.750);
\draw [dashed] (0, 1.500) node [left]{\tiny $2.5M$} -- (3.060,1.500);
\draw [dashed] (0, 2.250) node [left]{\tiny $3.75M$} -- (3.060,2.250);
\draw [fill=\typea,thin] (0.50,3.00) node[right] at(0.65,3.15) {\tiny Exact} rectangle (0.72,3.23)  ;
\draw [fill=\typeb,thin] (2.00,3.00) node[right]at(2.15,3.15) {\tiny 32K} rectangle (2.23,3.23)  ;
\end{tikzpicture}
  }
     \subfloat[youtube]{\label{fig:d-youtube}
    %
\begin{tikzpicture}[scale=1.000000]\draw [fill=\typea,thin] (0.100,0) node[below] at (0.250,0){{\tiny (a)}} rectangle (0.250,0.010) ;
\draw [fill=\typeb,thin] (0.260,0) rectangle (0.410,0.010) ;
\draw [fill=\typea,thin] (0.520,0) node[below] at (0.670,0){{\tiny (b)}} rectangle (0.670,0.000) ;
\draw [fill=\typeb,thin] (0.680,0) rectangle (0.830,0.000) ;
\draw [fill=\typea,thin] (0.940,0) node[below] at (1.090,0){{\tiny (c)}} rectangle (1.090,0.046) ;
\draw [fill=\typeb,thin] (1.100,0) rectangle (1.250,0.031) ;
\draw [fill=\typea,thin] (1.360,0) node[below] at (1.510,0){{\tiny (d)}} rectangle (1.510,0.012) ;
\draw [fill=\typeb,thin] (1.520,0) rectangle (1.670,0.013) ;
\draw [fill=\typea,thin] (1.780,0) node[below] at (1.930,0){{\tiny (e)}} rectangle (1.930,0.126) ;
\draw [fill=\typeb,thin] (1.940,0) rectangle (2.090,0.100) ;
\draw [fill=\typea,thin] (2.200,0) node[below] at (2.350,0){{\tiny (f)}} rectangle (2.350,0.418) ;
\draw [fill=\typeb,thin] (2.360,0) rectangle (2.510,0.480) ;
\draw [fill=\typea,thin] (2.620,0) node[below] at (2.770,0){{\tiny (g)}} rectangle (2.770,2.440) ;
\draw [fill=\typeb,thin] (2.780,0) rectangle (2.930,2.600) ;
\draw [<->, thick] (3.24,0) -- (0,0)-- (0,0.00) node[left] at (-0.8, 1.50) {\rotatebox{90}{\tiny \# Triangles}} -- (0, 3.00) ;
\draw [dashed] (0, 0.750) node [left]{\tiny $750K$} -- (3.060,0.750);
\draw [dashed] (0, 1.500) node [left]{\tiny $1.5M$} -- (3.060,1.500);
\draw [dashed] (0, 2.250) node [left]{\tiny $2.25M$} -- (3.060,2.250);
\draw [fill=\typea,thin] (0.50,3.00) node[right] at(0.65,3.15) {\tiny Exact} rectangle (0.72,3.23)  ;
\draw [fill=\typeb,thin] (2.00,3.00) node[right]at(2.15,3.15) {\tiny 32K} rectangle (2.23,3.23)  ;
\end{tikzpicture}
  }\\
\subfloat[flickr]{\label{fig:d-flickr}
    %
\begin{tikzpicture}[scale=1.000000]\draw [fill=\typea,thin] (0.100,0) node[below] at (0.250,0){{\tiny (a)}} rectangle (0.250,1.661) ;
\draw [fill=\typeb,thin] (0.260,0) rectangle (0.410,1.658) ;
\draw [fill=\typea,thin] (0.520,0) node[below] at (0.670,0){{\tiny (b)}} rectangle (0.670,0.014) ;
\draw [fill=\typeb,thin] (0.680,0) rectangle (0.830,0.015) ;
\draw [fill=\typea,thin] (0.940,0) node[below] at (1.090,0){{\tiny (c)}} rectangle (1.090,1.644) ;
\draw [fill=\typeb,thin] (1.100,0) rectangle (1.250,1.643) ;
\draw [fill=\typea,thin] (1.360,0) node[below] at (1.510,0){{\tiny (d)}} rectangle (1.510,0.403) ;
\draw [fill=\typeb,thin] (1.520,0) rectangle (1.670,0.458) ;
\draw [fill=\typea,thin] (1.780,0) node[below] at (1.930,0){{\tiny (e)}} rectangle (1.930,1.322) ;
\draw [fill=\typeb,thin] (1.940,0) rectangle (2.090,1.348) ;
\draw [fill=\typea,thin] (2.200,0) node[below] at (2.350,0){{\tiny (f)}} rectangle (2.350,2.738) ;
\draw [fill=\typeb,thin] (2.360,0) rectangle (2.510,2.663) ;
\draw [fill=\typea,thin] (2.620,0) node[below] at (2.770,0){{\tiny (g)}} rectangle (2.770,2.513) ;
\draw [fill=\typeb,thin] (2.780,0) rectangle (2.930,2.475) ;
\draw [<->, thick] (3.24,0) -- (0,0)-- (0,0.00) node[left] at (-0.8, 1.50) {\rotatebox{90}{\tiny \# Triangles}} -- (0, 3.00) ;
\draw [dashed] (0, 0.750) node [left]{\tiny $40M$} -- (3.060,0.750);
\draw [dashed] (0, 1.500) node [left]{\tiny $80M$} -- (3.060,1.500);
\draw [dashed] (0, 2.250) node [left]{\tiny $120M$} -- (3.060,2.250);
\draw [fill=\typea,thin] (0.50,3.00) node[right] at(0.65,3.15) {\tiny Exact} rectangle (0.72,3.23)  ;
\draw [fill=\typeb,thin] (2.00,3.00) node[right]at(2.15,3.15) {\tiny 32K} rectangle (2.23,3.23)  ;
\end{tikzpicture}
  }
\subfloat[livejournal]{\label{fig:d-livejournal}
    %
\begin{tikzpicture}[scale=1.000000]\draw [fill=\typea,thin] (0.100,0) node[below] at (0.250,0){{\tiny (a)}} rectangle (0.250,1.791) ;
\draw [fill=\typeb,thin] (0.260,0) rectangle (0.410,1.758) ;
\draw [fill=\typea,thin] (0.520,0) node[below] at (0.670,0){{\tiny (b)}} rectangle (0.670,0.006) ;
\draw [fill=\typeb,thin] (0.680,0) rectangle (0.830,0.007) ;
\draw [fill=\typea,thin] (0.940,0) node[below] at (1.090,0){{\tiny (c)}} rectangle (1.090,1.299) ;
\draw [fill=\typeb,thin] (1.100,0) rectangle (1.250,1.308) ;
\draw [fill=\typea,thin] (1.360,0) node[below] at (1.510,0){{\tiny (d)}} rectangle (1.510,0.219) ;
\draw [fill=\typeb,thin] (1.520,0) rectangle (1.670,0.224) ;
\draw [fill=\typea,thin] (1.780,0) node[below] at (1.930,0){{\tiny (e)}} rectangle (1.930,1.371) ;
\draw [fill=\typeb,thin] (1.940,0) rectangle (2.090,1.362) ;
\draw [fill=\typea,thin] (2.200,0) node[below] at (2.350,0){{\tiny (f)}} rectangle (2.350,2.034) ;
\draw [fill=\typeb,thin] (2.360,0) rectangle (2.510,2.094) ;
\draw [fill=\typea,thin] (2.620,0) node[below] at (2.770,0){{\tiny (g)}} rectangle (2.770,2.607) ;
\draw [fill=\typeb,thin] (2.780,0) rectangle (2.930,2.643) ;
\draw [<->, thick] (3.24,0) -- (0,0)-- (0,0.00) node[left] at (-0.8, 1.50) {\rotatebox{90}{\tiny \# Triangles}} -- (0, 3.00) ;
\draw [dashed] (0, 0.750) node [left]{\tiny $25M$} -- (3.060,0.750);
\draw [dashed] (0, 1.500) node [left]{\tiny $50M$} -- (3.060,1.500);
\draw [dashed] (0, 2.250) node [left]{\tiny $75M$} -- (3.060,2.250);
\draw [fill=\typea,thin] (0.50,3.00) node[right] at(0.65,3.15) {\tiny Exact} rectangle (0.72,3.23)  ;
\draw [fill=\typeb,thin] (2.00,3.00) node[right]at(2.15,3.15) {\tiny 32K} rectangle (2.23,3.23)  ;
\end{tikzpicture}
  }
   
\caption{Directed triangle counts using wedge sampling}
\label{fig:dtd}
\end{figure}
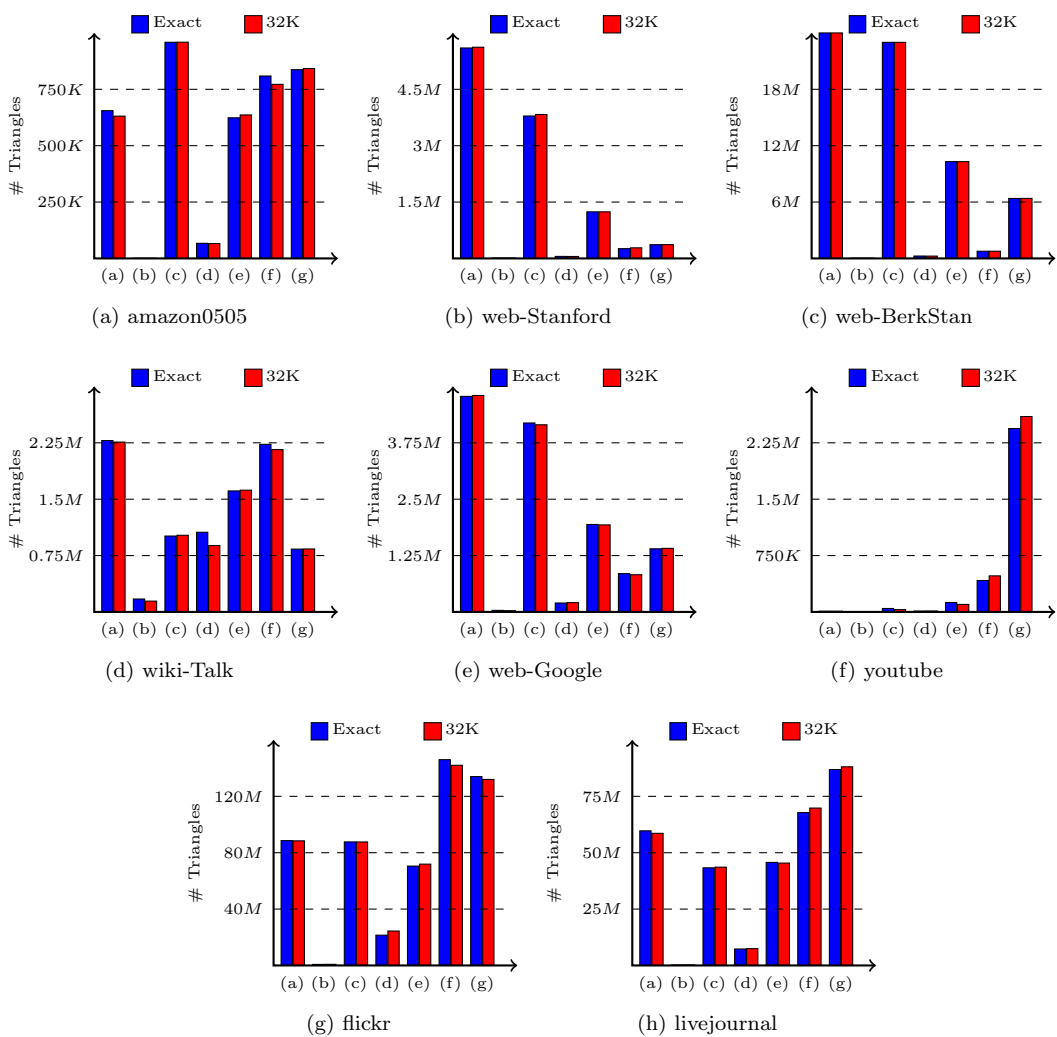

\begin{figure}[tbh]
\centering
\subfloat[amazon0505]{\label{fig:drel-amazon}
\begin{tikzpicture}\draw [fill=\typea,thin] (0.150,0) node[below] at (0.300,0){{\tiny (a)}} rectangle (0.450,1.111) ;
\draw [fill=\typea,thin] (0.610,0) node[below] at (0.760,0){{\tiny (b)}} rectangle (0.910,3.141) ;
\draw [fill=\typea,thin] (1.070,0) node[below] at (1.220,0){{\tiny (c)}} rectangle (1.370,0.014) ;
\draw [fill=\typea,thin] (1.530,0) node[below] at (1.680,0){{\tiny (d)}} rectangle (1.830,0.445) ;
\draw [fill=\typea,thin] (1.990,0) node[below] at (2.140,0){{\tiny (e)}} rectangle (2.290,0.601) ;
\draw [fill=\typea,thin] (2.450,0) node[below] at (2.600,0){{\tiny (f)}} rectangle (2.750,1.370) ;
\draw [fill=\typea,thin] (2.910,0) node[below] at (3.060,0){{\tiny (g)}} rectangle (3.210,0.186) ;
\draw [<->, thick] (3.67,0) -- (0,0)-- (0,0.09) node[left] at (-0.8, 1.80) {\rotatebox{90}{\scriptsize Relative error}} -- (0, 3.60) ;
\draw [dashed] (0, 0.900) node [left]{\scriptsize $0.03$} -- (3.390,0.900);
\draw [dashed] (0, 1.800) node [left]{\scriptsize$0.06$} -- (3.390,1.800);
\draw [dashed] (0, 2.700) node [left]{\scriptsize$0.09$} -- (3.390,2.700);
\end{tikzpicture}
}
\subfloat[web-Stanford]{\label{fig:drel-stanford}
\begin{tikzpicture}\draw [fill=\typea,thin] (0.150,0) node[below] at (0.300,0){{\tiny (a)}} rectangle (0.450,0.278) ;
\draw [fill=\typea,thin] (0.610,0) node[below] at (0.760,0){{\tiny (b)}} rectangle (0.910,2.872) ;
\draw [fill=\typea,thin] (1.070,0) node[below] at (1.220,0){{\tiny (c)}} rectangle (1.370,1.353) ;
\draw [fill=\typea,thin] (1.530,0) node[below] at (1.680,0){{\tiny (d)}} rectangle (1.830,1.739) ;
\draw [fill=\typea,thin] (1.990,0) node[below] at (2.140,0){{\tiny (e)}} rectangle (2.290,0.819) ;
\draw [fill=\typea,thin] (2.450,0) node[below] at (2.600,0){{\tiny (f)}} rectangle (2.750,1.277) ;
\draw [fill=\typea,thin] (2.910,0) node[below] at (3.060,0){{\tiny (g)}} rectangle (3.210,0.403) ;
\draw [<->, thick] (3.67,0) -- (0,0)-- (0,0.09) node[left] at (-0.8, 1.80) {\rotatebox{90}{\scriptsize Relative error}} -- (0, 3.60) ;
\draw [dashed] (0, 0.900) node [left]{\scriptsize $0.03$} -- (3.390,0.900);
\draw [dashed] (0, 1.800) node [left]{\scriptsize$0.06$} -- (3.390,1.800);
\draw [dashed] (0, 2.700) node [left]{\scriptsize$0.09$} -- (3.390,2.700);
\end{tikzpicture}
}
\subfloat[Average]{\label{fig:drel-average}
\begin{tikzpicture}\draw [fill=\typea,thin] (0.150,0) node[below] at (0.300,0){{\tiny (a)}} rectangle (0.450,0.294) ;
\draw [fill=\typea,thin] (0.610,0) node[below] at (0.760,0){{\tiny (b)}} rectangle (0.910,3.275) ;
\draw [fill=\typea,thin] (1.070,0) node[below] at (1.220,0){{\tiny (c)}} rectangle (1.370,1.361) ;
\draw [fill=\typea,thin] (1.530,0) node[below] at (1.680,0){{\tiny (d)}} rectangle (1.830,1.675) ;
\draw [fill=\typea,thin] (1.990,0) node[below] at (2.140,0){{\tiny (e)}} rectangle (2.290,0.988) ;
\draw [fill=\typea,thin] (2.450,0) node[below] at (2.600,0){{\tiny (f)}} rectangle (2.750,1.481) ;
\draw [fill=\typea,thin] (2.910,0) node[below] at (3.060,0){{\tiny (g)}} rectangle (3.210,0.422) ;
\draw [<->, thick] (3.67,0) -- (0,0)-- (0,0.09) node[left] at (-0.8, 1.80) {\rotatebox{90}{\scriptsize Relative error}} -- (0, 3.60) ;
\draw [dashed] (0, 0.900) node [left]{\scriptsize $0.03$} -- (3.390,0.900);
\draw [dashed] (0, 1.800) node [left]{\scriptsize$0.06$} -- (3.390,1.800);
\draw [dashed] (0, 2.700) node [left]{\scriptsize$0.09$} -- (3.390,2.700);
\end{tikzpicture}
}
\caption{Relative error for the 7 triangle types}
\label{fig:dsummary}
\end{figure}
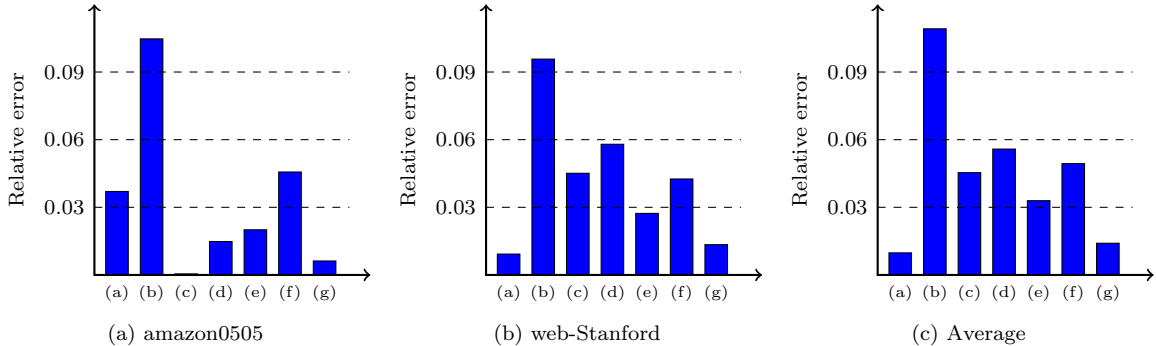

We take a closer look at the results, using \emph{relative error} plots. For brevity, we provide
these only for \texttt{amazon0505} and \texttt{web-Stanford} graphs, and give an average over all the graphs. A clear pattern 
appears, where the relative error for type-(b) triangles is worse than the others.
This is where a weakness of wedge sampling becomes apparent: When a triangle type is extremely infrequent, wedge
sampling is limited in the quality of the estimate. 
More details are shown in \Tab{typeb}, where we show the relative error for these
triangles, as well the fraction of these triangles (with respect to the total count). Across the board,
we see that type-(b) triangles are extremely rare; hence, wedge sampling is unable to get sharp estimates.
It would be very interesting to design a fast method that accuractely counts such rare triangles.

\begin{table}[thbp]
\caption{A closer look at type-(b) triangles}
\label{tab:typeb}
\centering
\begin{tabular}{|r|c|r|c|}
  \hline
Graph & Relative Error & Count & Ratio to all triangles \\ \hline 
amazon0505 &	0.105 & 623&	0.00013 \\
web-Stanford	& 0.061 & 9016& 	0.00053 \\
web-BerkStan	& 0.063 & 	23035 & 0.00026 \\
wiki-Talk&	0.162& 171903&	0.01497 \\
web-Google&	0.166& 34418&	0.00189 \\ 
flickr&	0.030& 773098 &	0.00111 \\
youtube&	0.113& 118 & 	0.00002 \\
livejournal&	0.173& 188666&	0.00047 \\\hline 
\end{tabular}
\end{table} 

There are some subtleties about the implementation that are worth mentioning. 
We have a choice of which wedge types to use to count $\ttype$-triangles. We can get some reuse if we use the same wedge type to count multiple triangles. For instance, type-\ttrans\  triangles can be counted using sampling over type \wout, \wmid, or \win  wedges.
By \Thm{main-dir}, sampling over $\psi$-wedges yields error is $\eps W_\psi/\chi(\psi,\ttype)$.
Hence, \emph{less frequent wedge types give better approximations}
for the same triangle type. Symmetrically,  the same wedge type can be used to count multiple triangle types. For instance, \wmid-wedges can be used to count type \ttrans, \tcycle, 
and \tmidrecip triangles.  
We use the following wedges for our triangle counts:
\begin{asparaitem}
  \item type-(ii) wedges are used to count triangles of type (a) and (b),
  \item type-(v) wedges are used to count triangles of type (c) and (f),
  \item type-(iv) wedges are used to count triangles of type (d) and (e), and
  \item type-(vi) wedges are used to count triangles of type (g).
\end{asparaitem}
The combinations that we use are boldface in \Fig{triangle_by_wedge}.
The choice of wedge can impact the accuracy of the directed triangle counts.
For example, type-(a)-triangles can be counted using wedges of type (i), (ii), or (iii). In \Fig{3wedges},
we show the relative error in the estimate for type-(a) triangles using 32K samples
via the different wedge types. Because of the high frequency of type-(iii) wedges,
the error in triangle counts using these wedges is quite large.
Other than one case, (ii)-wedges are the best choice. 

\begin{figure*}
\centering
\begin{tikzpicture}\draw [fill=\typea,thin] (0.150,0) rectangle (0.400,0.285) ;
\draw [fill=\typeb,thin] (0.410,0) node[below]{\rotatebox[origin=t]{90}{\scriptsize amazon0505}} rectangle (0.660,0.741) ;
\draw [fill=\typec,thin] (0.670,0) rectangle (0.920,0.141) ;
\draw [fill=\typea,thin] (1.080,0) rectangle (1.330,0.156) ;
\draw [fill=\typeb,thin] (1.340,0) node[below]{\rotatebox[origin=t]{90}{\scriptsize web-Stanford}} rectangle (1.590,0.054) ;
\draw [fill=\typec,thin] (1.600,0) rectangle (1.850,3.000) ;
\draw [fill=\typea,thin] (2.010,0) rectangle (2.260,0.224) ;
\draw [fill=\typeb,thin] (2.270,0) node[below]{\rotatebox[origin=t]{90}{\scriptsize web-BerkStan}} rectangle (2.520,0.026) ;
\draw [fill=\typec,thin] (2.530,0) rectangle (2.780,3.000) ;
\draw [fill=\typea,thin] (2.940,0) rectangle (3.190,3.000) ;
\draw [fill=\typeb,thin] (3.200,0) node[below]{\rotatebox[origin=t]{90}{\scriptsize wiki-Talk}} rectangle (3.450,0.187) ;
\draw [fill=\typec,thin] (3.460,0) rectangle (3.710,0.222) ;
\draw [fill=\typea,thin] (3.870,0) rectangle (4.120,0.163) ;
\draw [fill=\typeb,thin] (4.130,0) node[below]{\rotatebox[origin=t]{90}{\scriptsize web-Google}} rectangle (4.380,0.086) ;
\draw [fill=\typec,thin] (4.390,0) rectangle (4.640,1.104) ;
\draw [fill=\typea,thin] (4.800,0) rectangle (5.050,0.469) ;
\draw [fill=\typeb,thin] (5.060,0) node[below]{\rotatebox[origin=t]{90}{\scriptsize flickr-links}} rectangle (5.310,0.034) ;
\draw [fill=\typec,thin] (5.320,0) rectangle (5.570,0.291) ;
\draw [fill=\typea,thin] (5.730,0) rectangle (5.980,3.000) ;
\draw [fill=\typeb,thin] (5.990,0) node[below]{\rotatebox[origin=t]{90}{\scriptsize youtube-links}} rectangle (6.240,0.129) ;
\draw [fill=\typec,thin] (6.250,0) rectangle (6.500,3.000) ;
\draw [fill=\typea,thin] (6.660,0) rectangle (6.910,0.248) ;
\draw [fill=\typeb,thin] (6.920,0) node[below]{\rotatebox[origin=t]{90}{\scriptsize livejournal}} rectangle (7.170,0.372) ;
\draw [fill=\typec,thin] (7.180,0) rectangle (7.430,0.083) ;
\draw [<->, thick] (7.89,0) -- (0,0)-- (0,0.06) node[left] at (-0.8, 1.50) {\rotatebox{90}{\scriptsize Relative Error}} -- (0, 3.30) ;
\draw [dashed] (0, 1.000) node [left]{\scriptsize $0.05$} -- (7.910,1.000);
\draw [dashed] (0, 2.000) node [left]{\scriptsize$0.1$} -- (7.910,2.000);
\draw [dashed] (0, 3.000) node [left]{\scriptsize$0.15$} -- (7.910,3.000);
\draw [fill=\typea,thin] (1.40,3.30) node[right] at(2.15,3.55) {\small (i)} rectangle (1.90,3.80)  ;
\draw [fill=\typeb,thin] (3.86,3.30) node[right]at(4.61,3.55) {\small (ii)} rectangle (4.36,3.80)  ;
\draw [fill=\typec,thin] (6.32,3.30) node[right]at(7.07,3.55) {\small (iii)} rectangle (6.82,3.80)  ;
\end{tikzpicture}
\caption{Estimates of the number of type-(a) triangles using different wedges types. Sample size
is fixed to 32K.} 
\label{fig:3wedges} 
\end{figure*}
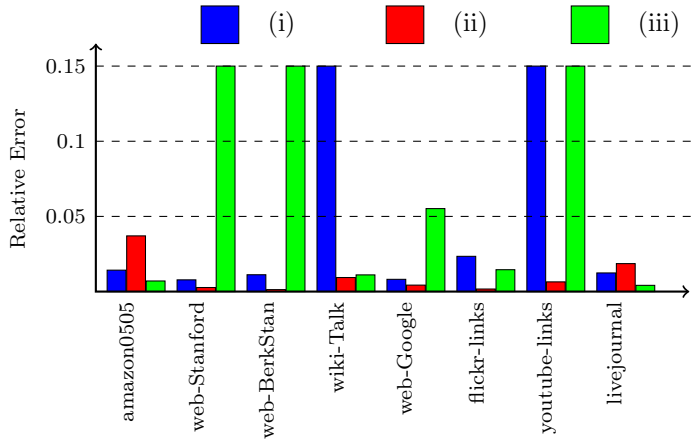

\section{Conclusions and future work} 
\label{sec:conc}
We proposed a series of wedge-based algorithms for computing
various triadic measures on graphs. Our algorithms come with
theoretical guarantees in the form of specific error and confidence
bounds.  The number of samples required to
achieve a specified error and confidence bound is independent of the
graph size.  For instance, 38,000 samples suffice for an additive error in the 
transitivity of less
than 0.01 with 99.9\% confidence \emph{for any graph}. The limiting
factors have to do with determining the sampling proportions; for instance, we
need to know the degree of each vertex and the overall degree
distribution to compute the transitivity.

The flexibility of wedge sampling along with the hard error bounds
essentially redefines what is feasible in terms of graph analysis. The
very expensive computation of clustering coefficient is now much
faster and we can
consider much larger graphs than before. In an extension of this work,
we described a MapReduce implementation of this method that scales
to billions of edges, leading to some of the first published results
at these sizes~\cite{KoPiPlSe13}.

With triadic analysis no longer being a computational burden, we can
extend our horizons into new territories and look at attributed triangles (e.g., we might compare the clustering
coefficient for ``teenage'' versus ``adult'' nodes in a social network),
evolution of triadic connections, higher-order structures such as
4-cycles and 4-cliques, and so on.

\bibliographystyle{siammod}

\end{document}